\PassOptionsToPackage{dvipsnames}{xcolor}
\PassOptionsToPackage{table}{xcolor}
\documentclass[a4paper]{article}

\usepackage{microtype}
\usepackage{graphicx}
\usepackage{tikz}
\usepackage{amsmath}
\usepackage{mathdots}
\usetikzlibrary{math, positioning, decorations.pathreplacing, backgrounds, decorations.shapes}
\usepackage{hyperref}

\usepackage{amsmath}%
\usepackage{amsthm}%
\usepackage{amssymb}%

\usepackage{booktabs}
\usepackage{xspace}%
\usepackage[linesnumbered, ruled, vlined, noend]{algorithm2e}
\usepackage[capitalize]{cleveref}%
\usepackage{thm-restate}%
\usepackage[numbers]{natbib}

\definecolor{Gray}{gray}{0.85}
\newcolumntype{a}{>{\columncolor{Gray}}r}

\theoremstyle{plain}%
\newtheorem{theorem}{Theorem}

\newtheorem{lemma}[theorem]{Lemma}

\newtheorem{corollary}[theorem]{Corollary}

\theoremstyle{plain}%
\newtheorem{definition}[theorem]{Definition}

\newcommand{\remove}[1]{}%

\newcommand{\ProbC}{{\mathbb{P}}}

\newcommand{\Prob}[1]{\ProbC\left[ #1 \right]}

\renewcommand{\Re}{\mathbb{R}}%
\newcommand{\reals}{\Re}%
\newcommand{\posreals}{\Re^+}%
\newcommand{\nat}{\mathbb{N}}%

\newcommand{\ultrametric}{\ensuremath{\Delta}}
\newcommand{\metric}{\ensuremath{\ell}}
\newcommand{\euclideanmetric}{\ensuremath{\ell_2}}

\newcommand{\enorm}[1]{\left\lVert #1 \right\rVert_2}

\newcommand{\BUF}{\ensuremath{\textsc{BUF}_\infty}}

\newcommand{\distortionratio}{\ensuremath{\textsc{dist}_\infty}}

\newcommand{\CutW}{\textsc{CW}}
\newcommand{\ACW}{\textsc{ACW}}
\newcommand{\cutw}[1]{\CutW(#1)}

\newcommand{\acutw}[1]{\ACW(#1)}

\newcommand{\lca}{\ensuremath{\textsc{lca}}}

\newcommand{\cO}{O}
\newcommand{\cOtilde}{\tilde{O}}

\newcommand{\afnpb}{\textsc{AFN}\xspace}

\newcommand{\Cc}{\mathcal{C}}
\newcommand{\Dd}{\mathcal{D}}
\newcommand{\Hh}{\mathcal{H}}

\newcommand{\iris}{\texttt{IRIS}}
\newcommand{\mice}{\texttt{MICE}}
\newcommand{\pendigits}{\texttt{PENDIGITS}}
\newcommand{\shuttle}{\texttt{SHUTTLE}}
\newcommand{\diabetes}{\texttt{DIABETES}}

\newcommand{\fastbuf}{\textsc{FastUlt}\xspace}

\usepackage{color}
\usepackage{fancyvrb}

\title{A $(1+\epsilon)$-Approximation for Ultrametric Embedding in Subquadratic Time}

\author{
    Gabriel Bathie\thanks{LaBRI, University of Bordeaux, France. \url{https://perso.ens-lyon.fr/gabriel.bathie/}} \and
    Guillaume Lagarde\thanks{LaBRI, University of Bordeaux, France. \url{https://guillaume-lagarde.github.io/}}
    }

\begin{document}
\maketitle
\begin{abstract}
    Efficiently computing accurate representations of high-dimensional data is essential for data analysis and unsupervised learning. Dendrograms, also known as ultrametrics, are widely used representations that preserve hierarchical relationships within the data. However, popular methods for computing them, such as \emph{linkage} algorithms, suffer from quadratic time and space complexity, making them impractical for large datasets. 
    The ``best ultrametric embedding'' (a.k.a. ``best ultrametric fit'') problem, which aims to find the ultrametric that best preserves the distances between points in the original data, is known to require at least quadratic time for an exact solution.
    Recent work has focused on improving scalability by approximating optimal solutions in subquadratic time, resulting in a $(\sqrt{2} + \epsilon)$-approximation (Cohen-Addad, de Joannis de Verclos and Lagarde, 2021).

    In this paper, we present the first subquadratic algorithm that achieves arbitrarily precise approximations of the optimal ultrametric embedding. Specifically, we provide an algorithm that, for any $c \geq 1$, outputs a $c$-approximation of the best ultrametric in time $\tilde{O}(n^{1 + 1/c})$. In particular, for any fixed $\epsilon > 0$, the algorithm computes a $(1+\epsilon)$-approximation in time $\tilde{O}(n^{2 - \epsilon + o(\epsilon^2)})$.

    Experimental results show that our algorithm improves upon previous methods in terms of approximation quality while maintaining comparable running times.
\end{abstract}


\section{Introduction}\label{sec:intro}

Clustering is a fundamental technique in data analysis that is used to group similar data points. It helps in uncovering underlying patterns, segmenting data into meaningful categories, and simplifying data representation. Applications of clustering span various domains, including for example bioinformatics, market segmentation, social network analysis, image processing, feature learning, spatial and geoscience, and many others; we refer to \cite{MurtaghC17} for a comprehensive list of references.

While 'flat' clustering methods, such as k-means, can effectively partition data into distinct groups, they often fall short when dealing with complex data structures, in particular data points that exhibit multiscale structures. These methods typically assume a fixed number of clusters or rely on specific distance thresholds, which may not capture the true nature of the data. Consequently, they can struggle to reveal the intricate relationships between data points.

Hierarchical clustering addresses these limitations by building a multi-level hierarchy of clusters. This approach does not require specifying the number of clusters a priori and can reveal the nested structure of the data at all levels of granularity. Hierarchical clustering iteratively splits or merges data points based on similarity, forming a tree-like representation known as a dendrogram. This tree structure provides a comprehensive view of the data's organization, allowing for more nuanced interpretations and insights.

Dendrograms and hierarchical clusterings are formalized and quantified using the mathematical concept of \emph{ultrametric}\footnote{An ultrametric is a metric that satisfies a stronger triangle inequality: for any three points, the distance between any two points is at most the maximum of the distances between the other two pairs.}. Ultrametrics provides a rigorous foundation for hierarchical clustering and allows the development and analysis of efficient algorithms to compute and analyze the hierarchical structure of data, see for example \cite{jain1988algorithms} or \cite{CarlssonM10}. 

The most popular methods for constructing an ultrametric are the \emph{agglomerative algorithms}, such as single linkage, complete linkage, average linkage, and Ward's method. These algorithms work \emph{bottom-up} and build an ultrametric by iteratively merging the closest clusters based on a distance metric. While widely used and successful in many applications, they suffer from two major limitations: first, it is not clear what objective functions these methods aim to optimize, making them difficult to analyze and lacking approximation guarantees; second, they have at least quadratic time and space complexity, making them impractical for handling large datasets.

In this paper, we consider the problem of computing the best ultrametric fit ($\BUF$), which is quantified by how well an ultrametric preserves the distances of the original metric using the notion of \emph{maximum distortion}. Formally, given a set \(X\) of points in Euclidean space, our goal is to find an ultrametric \(\ultrametric\) such that for all \(x, y \in X\), \(\enorm{x - y} \leq \ultrametric(x, y) \leq c \cdot \enorm{x - y}\), where \(c\) is as small as possible. This problem was first introduced by \citet{FKW95}, who proved that it cannot be solved in subquadratic time and provided an exact algorithm matching this lower bound.

To overcome this quadratic time barrier, \citet{CKL20} and \citet{CDL21} initiated the development of faster algorithms for computing \emph{approximations} of the optimal ultrametric. They proposed respectively algorithms achieving a \(5c\)-approximation and a \(\sqrt{2}c\)-approximation in time $\cOtilde(n^{1 + 12/c^2})$, thus providing subquadratic algorithms to approximate $\BUF$ up to factors of \(\approx 17.32\) and \(\approx 4.90\), respectively.

\subsection{Our contribution}
The fundamental question explored in this paper is whether \textbf{we can achieve arbitrarily precise approximations of the optimal ultrametric fit in subquadratic time}, or if there exists a theoretical barrier preventing this.

We answer this question positively by constructing the first subquadratic algorithm that achieves arbitrarily precise approximations of the optimal solution to $\BUF$. More precisely, we prove the following theorem:

\begin{theorem}\label{thm:main}
    For any $\gamma \geq 1$ and $\alpha > 1$, there exists an algorithm that computes a $\gamma \cdot \alpha$-approximation of $\BUF$ in time $\cOtilde(n^{1 + 1/\gamma^2} + n^{1 + 1/\alpha^2})$ and space $\cOtilde(n^{1 + 1/\gamma^2} + n^{1 + 1/\alpha^2})$.
\end{theorem}

Previously, the best known subquadratic time algorithm for \BUF, by \citet{CDL21}, could only handle approximation factors greater than $\sqrt{2} \cdot \sqrt{12} \approx 4.90$, while we can now achieve arbitrarily precise approximations. Our algorithm is based on two main components:

\begin{itemize}
    \item An algorithm\footnote{An algorithm was proposed in \cite{CKL20} and \cite{CDL21}, but it operates in subquadratic time only when \(\gamma \geq \sqrt{12}\)} for computing a \(\gamma\)-Kruskal Tree (abbreviated \(\gamma\)-KT) in time \(\tilde{\mathcal{O}}(n^{1+1/\gamma^2})\). This algorithm might be of independent interest. Unlike previous algorithms for $\gamma$-KT, it does not require the construction of a \(\gamma\)-spanner, which cannot be built in subquadratic time for \(\gamma < \sqrt{2}\) \citep{andoni2023sub}.
    \item A new data structure that computes an \(\alpha\)-approximation of the so-called cut weights in time \(\tilde{\mathcal{O}}(n^{1+1/\alpha^2})\) (See \cref{def:cw} for the definition of the cut weights). The previous best subquadratic time algorithm was only able to handle approximation factors greater than $\sqrt 2$, see~\cite{CDL21}. Our data structure is based on a dynamic version of the approximate farthest neighbor data structure developed by \citet{pagh2017approximate}.
\end{itemize}

\subsection{Experimental results}

To complement our theoretical results and to demonstrate the practical efficiency of our algorithm, we perform an extensive set of experiments.
We measure the performance of our algorithm both in terms of approximation factor and running time on five classical real-world datasets,
and evaluate its scalability on large synthetic datasets.
We compare our algorithm with the state-of-the-art algorithm of \citet{CDL21}
and the widely used implementation of the \texttt{fastcluster} Python package.
The results show that our algorithm yields better approximations than the algorithm of \citet{CDL21} while maintaining a comparable running time and that it can scale to datasets containing millions of points. 

\subsection{Related work}

\citet{CarlssonM10} established important foundations for hierarchical clustering, in particular through an in-depth study of ultrametrics that revealed the theoretical properties of hierarchical clustering algorithms.

Other works have explored the complexity of optimizing other distortion measures, such as the average distortion ($\ell_1$ norm) and more generally, the \(\ell_p\) norm for different values of \(p\). The problem is NP-complete for \(p = 1, 2\) and APX-hard for \(p = 1\) (see \cite{wareham1993complexity} and \cite{agarwala1998approximability}). \citet{ailon2011fitting} investigated the case of the \(\ell_p\) norm for various values of \(p\) and provided polynomial-time algorithms to \(O((\log n \log \log n)^{1/p})\)-approximate both the best ultrametric and tree metric embeddings using an LP formulation and rounding techniques. These studies consider problems that are NP-hard (at least for $p=1,2$) and provide approximation algorithms, but do not focus on scalability: their algorithms have a complexity of \(\Omega(n^4)\).

Subquadratic time algorithms in high-dimensional settings, such as those in our work, have been studied by \cite{gilpin2013efficient} and \cite{cochez2015twister}, who provide near-linear time algorithms in the best running case. However, these algorithms lack approximation guarantees for their outputs.

Another line of research focuses on different objective functions to quantify the quality of hierarchical clustering. For instance, \citet{dasgupta2016cost} introduced an objective function based on cluster properties. Since then, numerous efforts have been devoted to developing algorithms that optimize or approximate this and related metrics, see e.g. \cite{roy2017hierarchical}, \cite{charikar2017approximate} and \cite{cohen2017hierarchical}.

Significant work has also been done to understand the guarantees provided by popular algorithms. Recently, \cite{cohen2019hierarchical} and \cite{moseley2023approximation} proved that average linkage has a small constant approximation ratio for the dual of Dasgupta's objective function, while other methods, such as the bisecting $k$-means top-down approach, perform poorly for the same objective.

\subsection{Organization of the Paper}
First, we introduce the necessary definitions and notations in \cref{sec:preliminaries}. In \cref{sec:high-level-algo} we present the high-level algorithm of \cite{CKL20} and \cite{CDL21}, on which our work is based, and highlight our improvements. Next, we describe in detail our algorithms to compute a $\gamma$-KT and an $\alpha$-approximation of the cut weights for any $\gamma \geq 1, \alpha \geq 1$ in \cref{sec:kt} and \cref{sec:cw} respectively. Finally, we discuss our experimental results in detail in \cref{sec:experiments}.

\section{Preliminaries}\label{sec:preliminaries}
We use the standard complexity definition of $\cOtilde(\cdot)$, which hides polylogarithmic factors in the input size. Formally, a function $f(n)$ is in $\cOtilde(g(n))$ if there exists a constant $c > 0$ such that $f(n) \leq c \cdot g(n) \cdot \log^k n$ for some $k \in \nat$ and all large enough $n \in \nat$. 

For any dimension $d \in \nat$, we denote by $\euclideanmetric: \reals^d \times \reals^d \rightarrow \posreals$ the Euclidean distance between two points $x, y \in \reals^d$.
Given a set $X$, an ultrametric distance $\ultrametric$ on $X$ is a distance function $\ultrametric: X \times X \rightarrow \posreals$ where the triangle inequality is replaced by the stronger ultrametric inequality:
\[\forall x, y, z \in X, \ultrametric(x, z) \leq \max(\ultrametric(x, y), \ultrametric(y, z)).\]
In this case, $(X, \ultrametric)$ is an ultrametric space.
When $X$ is a finite set, an ultrametric can be represented by a tree $T$ together with a weight function $w: T \rightarrow \posreals$, such that
\begin{itemize}
    \item each element of $X$ is assigned to a leaf of $T$,
    \item the weight of each leaf is 0,
    \item the weights are decreasing along any path from the root to a leaf.
\end{itemize}

The ultrametric distance induced by $T$ and $w$ is denoted by $\ultrametric_T$ and is defined as $\ultrametric_T(u, v) = w(\lca(u, v))$, where $\lca(u, v)$ represents the least common ancestor of $u$ and $v$ in $T$. Using a tree representation of an ultrametric is useful for visualization and interpretation. 
Specifically, the tree structure can be viewed as a hierarchical clustering of data points, with each node in the tree representing a cluster formed by the leaves of the subtree rooted at that node. A cut through the tree corresponds to a partitioning of the data points into clusters, with cuts made at different levels of the tree yielding clusterings of different granularities. See \cref{fig:combined-ultrametrics} for an illustration. This representation has significant utility and numerous applications; see for example the extensive survey of \citet{MurtaghC17}.

  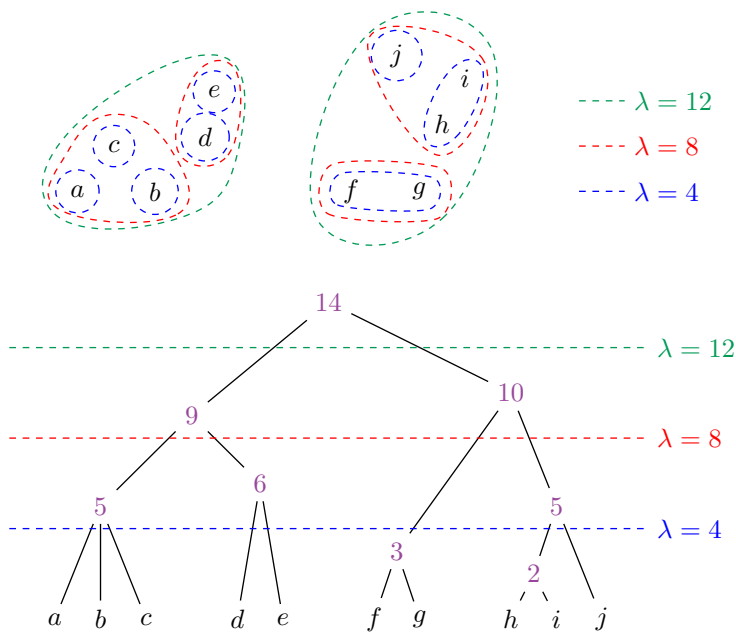
\begin{figure}[ht!]
    \centering
    \begin{center}
        \begin{tikzpicture}[scale=0.3]
          \tikzmath{
              \xscale = 2;
              \yscale = 2;
              \lwidth = 2pt * 0.25; 
              \dashon = 10 * 0.25;
              \dashoff = 10*0.25;
          }
            \node[draw, dash pattern={on \dashon pt off \dashoff pt}, line width=\lwidth, blue, circle] (a) at (0*\xscale, 0*\yscale) {$\color{black}a$};
            \node[draw, dash pattern={on \dashon pt off \dashoff pt}, line width=\lwidth, blue, circle] (b) at (1.7*\xscale, 0*\yscale) {$\color{black}b$};
            \node[draw, dash pattern={on \dashon pt off \dashoff pt}, line width=\lwidth, blue, circle] (c) at (0.8*\xscale, 1*\yscale) {$\color{black}c$};
            
            \node[draw, dash pattern={on \dashon pt off \dashoff pt}, line width=\lwidth, blue, circle] (d) at (2.8*\xscale, 1.2*\yscale) {$\color{black}d$};
            \node[draw, dash pattern={on \dashon pt off \dashoff pt}, line width=\lwidth, blue, circle] (e) at (3*\xscale, 2.2*\yscale) {$\color{black}e$};
            
            \node[] (f) at (6*\xscale, 0*\yscale) {$\color{black}f$};
            \node[] (g) at (7.5*\xscale, 0*\yscale) {$\color{black}g$};
            
            \node[] (h) at (8*\xscale, 1.5*\yscale) {$\color{black}h$};
            \node[] (i) at (8.5*\xscale, 2.5*\yscale) {$\color{black}i$};
            \node[draw, dash pattern={on \dashon pt off \dashoff pt}, line width=\lwidth, blue, circle] (j) at (7*\xscale, 3*\yscale) {$\color{black}j$};
      
            \draw[line width=\lwidth, dash pattern={on \dashon pt off \dashoff pt}, blue] (11*\xscale, 0*\yscale) -- (12*\xscale, 0*\yscale) node [right] {$\lambda = 4$};
            \draw[line width=\lwidth, dash pattern={on \dashon pt off \dashoff pt}, red] (11*\xscale, 1*\yscale) -- (12*\xscale, 1*\yscale) node [right] {$\lambda = 8$};
            \draw[line width=\lwidth, dash pattern={on \dashon pt off \dashoff pt}, ForestGreen] (11*\xscale, 2*\yscale) -- (12*\xscale, 2*\yscale) node [right] {$\lambda = 12$};
      
            \draw[line width=\lwidth, dash pattern={on \dashon pt off \dashoff pt}, red] plot [smooth cycle, tension=.8] coordinates {(-.4*\xscale, -.4*\yscale) (2.2*\xscale, -.5*\yscale) (2*\xscale, 1*\yscale) (0.8*\xscale, 1.7*\yscale) (-.2*\xscale, 0.9*\yscale)}; 
      
            \draw[line width=\lwidth, dash pattern={on \dashon pt off \dashoff pt}, red] plot [smooth cycle, tension=.8] coordinates {(2.2*\xscale, .9*\yscale) (3.3*\xscale, .8*\yscale) (3.5*\xscale, 2.6*\yscale) (2.6*\xscale, 2.6*\yscale)}; 
      
            \draw[line width=\lwidth, dash pattern={on \dashon pt off \dashoff pt}, blue] plot [smooth cycle] coordinates {(5.8*\xscale, .4*\yscale) (5.7*\xscale, -.3*\yscale) (7.7*\xscale, -.4*\yscale) (7.8*\xscale, .3*\yscale)}; 
            \draw[line width=\lwidth, dash pattern={on \dashon pt off \dashoff pt}, red] plot [smooth cycle] coordinates {(5.6*\xscale, .6*\yscale) (5.5*\xscale, -.5*\yscale) (7.9*\xscale, -.6*\yscale) (8*\xscale, .5*\yscale)}; 
      
            \draw[line width=\lwidth, dash pattern={on \dashon pt off \dashoff pt}, blue] plot [smooth cycle, tension=.9] coordinates {(7.6*\xscale, 1.3*\yscale) (8.3*\xscale, 1.2*\yscale) (8.9*\xscale, 2.6*\yscale) (8.2*\xscale, 2.7*\yscale)}; 
      
            \draw[line width=\lwidth, dash pattern={on \dashon pt off \dashoff pt}, red] plot [smooth cycle, tension=.9] coordinates {(7.9*\xscale, .8*\yscale) (8.9*\xscale, 2.9*\yscale) (6.4*\xscale, 3.4*\yscale)}; 
      
            \draw[line width=\lwidth, dash pattern={on \dashon pt off \dashoff pt}, ForestGreen] plot [smooth cycle, tension=.8] coordinates {(0*\xscale, 1.5*\yscale) (-.5*\xscale, -.5*\yscale) (2.2*\xscale, -.6*\yscale) (3.5*\xscale, .8*\yscale) (3.1*\xscale, 3*\yscale)}; 
      
            \draw[line width=\lwidth, dash pattern={on \dashon pt off \dashoff pt}, ForestGreen] plot [smooth cycle, tension=.8] coordinates {(5.2*\xscale, -.5*\yscale) (8*\xscale, -.6*\yscale)  (9.1*\xscale, 3.1*\yscale) (6.4*\xscale, 3.5*\yscale)}; 
      
        \end{tikzpicture}            
        \end{center}
    \begin{center}
        \tikzmath{
          \linewidth=0.5pt;
          \dashon = 10 * 0.25; 
          \dashoff = 10*0.25;
        }
        \begin{tikzpicture}[scale=0.3]
            \node (a) at (0*2, 0) {$a$};
            \node (b) at (1*2, 0) {$b$};
            \node (c) at (2*2, 0) {$c$};
            
            \node (d) at (4*2, 0) {$d$};
            \node (e) at (5*2, 0) {$e$};
            
            \node (f) at (7*2, 0) {$f$};
            \node (g) at (8*2, 0) {$g$};
            
            \node (h) at (10*2, 0) {$h$};
            \node (i) at (11*2, 0) {$i$};
            \node (j) at (12*2, 0) {$j$};
        
            \node (abc) at (1*2, 5) {\color{Purple}5};
            \draw[line width=\linewidth, black] (a) -- (abc);
            \draw[line width=\linewidth, black] (b) -- (abc);
            \draw[line width=\linewidth, black] (c) -- (abc);
            \node (de) at (4.5*2, 6) {\color{Purple}6};
            \draw[line width=\linewidth, black] (d) -- (de);
            \draw[line width=\linewidth, black] (e) -- (de);
            \node (fg) at (7.5*2, 3) {\color{Purple}3};
            \draw[line width=\linewidth, black] (f) -- (fg);
            \draw[line width=\linewidth, black] (g) -- (fg);
            \node (hi) at (10.5*2, 2) {\color{Purple}2};
            \draw[line width=\linewidth, black] (h) -- (hi);
            \draw[line width=\linewidth, black] (i) -- (hi);
            \node (hij) at (11*2, 5) {\color{Purple}5};
            \draw[line width=\linewidth, black] (hi) -- (hij);
            \draw[line width=\linewidth, black] (j) -- (hij);
        
            \node (abcde) at (3*2, 9) {\color{Purple}9};
            \draw[line width=\linewidth, black] (abc) -- (abcde);
            \draw[line width=\linewidth, black] (de) -- (abcde);
            \node (fghij) at (10*2, 10) {\color{Purple}10};
            \draw[line width=\linewidth, black] (fg) -- (fghij);
            \draw[line width=\linewidth, black] (hij) -- (fghij);
            
            \node (root) at (6*2, 14) {\color{Purple}14};
            \draw[line width=\linewidth, black] (abcde) -- (root);
            \draw[line width=\linewidth, black] (fghij) -- (root);
        
            \draw[line width=\linewidth, dash pattern={on \dashon pt off \dashoff pt}, blue] (-1*2, 4) -- (13*2, 4) node [right] {$\lambda = 4$};
            \draw[line width=\linewidth, dash pattern={on \dashon pt off \dashoff pt}, red] (-1*2, 8) -- (13*2, 8) node [right] {$\lambda = 8$};
            \draw[line width=\linewidth, dash pattern={on \dashon pt off \dashoff pt}, ForestGreen] (-1*2, 12) -- (13*2, 12) node [right] {$\lambda = 12$};
        \end{tikzpicture}            
        \end{center}
          \caption{Points and clusters at three different levels of granularity (top) and the corresponding dendrogram (bottom).}
          \label{fig:combined-ultrametrics}
\end{figure}

\paragraph{Distortion ratio.}
To quantify how well an ultrametric $\ultrametric$ preserves the distances of a metric $\metric$, we use the standard concept of \emph{distortion}. 
Given a metric space $(X, \metric)$ and an ultrametric $\ultrametric$ on $X$ such that $\metric(x, y) \leq \ultrametric(x, y)$ for all $(x, y) \in X^2$, the distortion of $\ultrametric$, denoted by $\distortionratio$,\footnote{Similarly, one can define, for any $p \in [1,\infty[$, $\textsc{dist}_p$ as the $p$-norm of the vector $\left(\frac{\ultrametric(u, v)}{\metric(u, v)}\right)_{(u, v) \in X^2}$. The parameter $p$ has a significant influence on the complexity of the problem. For example, when $p=2$, the best ultrametric fit problem becomes NP-hard, while for $p=\infty$ the problem can be solved in polynomial time.}
is the maximum ratio between the Euclidean distance and the ultrametric distance, i.e.,
\[
\distortionratio = \max_{(x, y) \in X^2, x \neq y} \frac{\ultrametric(x, y)}{\metric(x, y)}.
\]

\paragraph{Best ultrametric fit.}
As defined by \citet{FKW95}, the best ultrametric fit problem ($\BUF$) consists in finding, given a metric space $(X, \metric)$, an ultrametric $\ultrametric$ on $X$ that preserves the distances as well as possible. Formally, the problem is to find an ultrametric $\ultrametric$ such that:
\begin{itemize}
    \item $\metric(x, y) \leq \ultrametric(x, y) $ for all $(x, y) \in X^2$,
    \item the distortion ratio $\distortionratio(\metric, \ultrametric)$ is minimized.
\end{itemize}
We denote by $\distortionratio^*$ the (unique) distortion ratio of an optimal solution.

$\ultrametric$ is a \emph{$c$-approximation} of the best ultrametric fit if the distortion of $\ultrametric$ is at most $c$ times the distortion of the optimal ultrametric, or equivalently if for any pair of points $x, y \in X$, we have:
\[
\metric(x, y) \leq \ultrametric(x,y) \leq c \cdot \distortionratio^* \cdot \metric(x, y).
\]
The scalar $c$ is called the approximation factor of $\ultrametric$.
By extension, an algorithm is a $c$-approximation algorithm of $\BUF$ if it outputs an ultrametric which is a $c$-approximation of the optimal ultrametric embedding.

\paragraph{Working in Euclidean spaces.} 
From now on, and as in \cite{CKL20} and \cite{CDL21}, we consider the case where $X$ is a set of $n$ points in a Euclidean space $\mathbb{R}^d$ equipped with the Euclidean metric $\euclideanmetric$, which is one of the most natural settings for many applications in data analysis and unsupervised learning. Note that by the Johnson-Lindenstrauss lemma~\cite{johnson1986extensions}, the dimension $d$ can always be reduced to $O(\log n)$ while preserving the distances between pairs of points up to a multiplicative factor of $(1+\epsilon)$ for any fixed $\epsilon > 0$.

\section{High level algorithm from \cite{CKL20}}\label{sec:high-level-algo}

We build our approximation algorithm using the framework of \citet{CKL20}, who provide a way to compute a $5\cdot \gamma$-approximation of $\BUF$ in time $O(nd+ n^{1+12/\gamma^2})$. This result was later improved by \citet{CDL21} who provide, for any fixed $\epsilon > 0$, a \((\sqrt{2} + \epsilon) \cdot \gamma\)-approximation for the same asymptotic running time.
We briefly recall the main ideas of their approach and pinpoint where we improve upon it. We first need a few more definitions.

\begin{definition}\label{def:gammakt}
    Let $G =(V,E,w)$ be a weighted graph and let $\gamma \ge 1$.
    A spanning tree $T = (V, E_T)$ of $G$ is a $\gamma$-Kruskal tree (or $\gamma$-KT for short) of $G$ if
    for every edge $e\in E\setminus E_T$, we have
    \[w(e) \ge \frac{1}{\gamma} \max_{e'\in P_T(e)} w(e'),\]
    where $P_T(e)$ is the unique path in $T$ from one endpoint of $e$ to the other.
\end{definition}

Let $T$ be a spanning tree of the complete graph induced by the set of points $X$. For an edge $e = (x, y) \in T$, we denote by $L(e)$ (respectively $R(e)$) the connected component of $T \setminus \{e' \in T \mid \euclideanmetric(e') \leq \euclideanmetric(e)\}$ that contains $x$ (respectively $y$).

\begin{definition}\label{def:cw}
    The \emph{cut weight} $\cutw{e}$ of an edge $e$ is defined as the maximal distance between a point in $L(e)$ and a point in $R(e)$:
    \[\cutw{e} = \max_{x \in L(e), y \in R(e)} \euclideanmetric(x, y).\]
\end{definition}
The function $\CutW: e\mapsto \cutw{e}$ for $e\in E_T$ is referred to as the \emph{cut weights of $T$}. See \cref{fig:cw}.

\begin{figure}
\begin{center}
    \tikzmath{
      \xscale=0.6;
      \yscale=0.6;
      \lwidth=0.4pt;
      \isize=1.6pt;
      \dashon = 10 * 0.25; 
      \dashoff = 10*0.25; 
      \sX=10;
      \eX=11;
    }
    \begin{tikzpicture}
        \node (u) at (0*\xscale, 0*\yscale) {$u$};
        \node (v) at (2*\xscale, 0*\yscale) {$v$};
        \node[draw, circle, inner sep=\isize, fill=ForestGreen!50] (x1) at (-2*\xscale, -.2*\yscale) {};
        \node[draw, circle, inner sep=\isize, fill=ForestGreen!50] (x2) at (-1*\xscale, -.9*\yscale) {};
        \node[draw, circle, inner sep=\isize, fill=ForestGreen!50] (x3) at (-2.5*\xscale, .5*\yscale) {};
        \node[draw, circle, inner sep=\isize, fill=ForestGreen!50] (x4) at (-2.7*\xscale, -.9*\yscale) {};
        \node[draw, circle, inner sep=\isize, fill=ForestGreen!50] (x5) at (-4*\xscale, -1*\yscale) {};
        \node[draw, circle, inner sep=\isize, fill=ForestGreen!50] (x6) at (-2.3*\xscale, -1.8*\yscale) {};

        \draw[line width=\lwidth, Purple] (u) -- (v) node[midway, above] {$e$};
        \draw[line width=\lwidth, black] (u) -- (x1);
        \draw[line width=\lwidth, black] (x1) -- (x2);
        \draw[line width=\lwidth, black] (x1) -- (x3);
        \draw[line width=\lwidth, black] (x1) -- (x4);
        \draw[line width=\lwidth, black] (x4) -- (x5);
        \draw[line width=\lwidth, black] (x4) -- (x6);
        
        \draw[line width=\lwidth, red] (x3) -- ++(0.5*\xscale, 1*\yscale) node[draw, circle, black, inner sep=\isize, fill=black!40] {};
        \draw[line width=\lwidth, red] (x3) -- ++(-0.5*\xscale, 1*\yscale) node[draw, circle, black, inner sep=\isize, fill=black!40] {};
        \draw[line width=\lwidth, red] (x2) -- ++(1*\xscale, -0.2*\yscale) node[draw, circle, black, inner sep=\isize, fill=black!40] {};
        \draw[line width=\lwidth, red] (x2) -- ++(0.5*\xscale, -1*\yscale) node[draw, circle, black, inner sep=\isize, fill=black!40] {};
        \draw[line width=\lwidth, red] (x5) -- ++(-1*\xscale, 0.5*\yscale) node[draw, circle, black, inner sep=\isize, fill=black!40] {};
        \draw[line width=\lwidth, red] (x5) -- ++(-1*\xscale, -0.5*\yscale) node[draw, circle, black, inner sep=\isize, fill=black!40] {};
        \draw[line width=\lwidth, red] (x6) -- ++(0.5*\xscale, -1*\yscale) node[draw, circle, black, inner sep=\isize, fill=black!40] {};
        \draw[line width=\lwidth, red] (x6) -- ++(-0.5*\xscale, -1*\yscale) node[draw, circle, black, inner sep=\isize, fill=black!40] {};

        \draw[line width=\lwidth, dash pattern={on \dashon pt off \dashoff pt}, blue] (\sX*\xscale, 0*\yscale) -- (\eX*\xscale, 0*\yscale) node [right] {$R(e)$};
        \draw[line width=\lwidth, dash pattern={on \dashon pt off \dashoff pt}, ForestGreen] (\sX*\xscale, 1*\yscale) -- (\eX*\xscale, 1*\yscale) node [right] {$L(e)$};
        \draw[line width=\lwidth, red] (\sX*\xscale, -1*\yscale) -- (\eX*\xscale, -1*\yscale) node [right] {$\color{black}{\color{red} e'}: w({\color{red} e'}) \ge w({\color{Purple}e})$};

        \draw[line width=\lwidth, dash pattern={on \dashon pt off \dashoff pt}, ForestGreen] plot [smooth cycle, tension=.8] coordinates {(-.5*\xscale, -1.2*\yscale) (0.5*\xscale, 0.3*\yscale) (-2.5*\xscale, 1*\yscale) (-4.5*\xscale, -1*\yscale) (-2.3*\xscale, -2.3*\yscale)};

        \node[draw, circle, inner sep=\isize, fill=blue!50] (y1) at (3.5*\xscale,  .7*\yscale) {};
        \node[draw, circle, inner sep=\isize, fill=blue!50] (y2) at (3.5*\xscale, -.7*\yscale) {};
        \node[draw, circle, inner sep=\isize, fill=blue!50] (y3) at (4.4*\xscale, 1.7*\yscale) {};
        \node[draw, circle, inner sep=\isize, fill=blue!50] (y4) at (4.7*\xscale, 1.0*\yscale) {};
        \node[draw, circle, inner sep=\isize, fill=blue!50] (y5) at (4.5*\xscale, -.7*\yscale) {};
        \node[draw, circle, inner sep=\isize, fill=blue!50] (y6) at (5.5*\xscale,  -0.2*\yscale) {};
        \node[draw, circle, inner sep=\isize, fill=blue!50] (y7) at (5.5*\xscale, -1.2*\yscale) {};          
        \draw[line width=\lwidth, black] (v) -- (y1);
        \draw[line width=\lwidth, black] (v) -- (y2);
        \draw[line width=\lwidth, black] (y1) -- (y3);
        \draw[line width=\lwidth, black] (y1) -- (y4);
        \draw[line width=\lwidth, black] (y2) -- (y5);
        \draw[line width=\lwidth, black] (y5) -- (y6);
        \draw[line width=\lwidth, black] (y5) -- (y7);

        \draw[line width=\lwidth, red] (y1) -- ++(-0.5*\xscale, 1*\yscale)  node[draw, circle, black, inner sep=\isize, fill=black!40] {};
        \draw[line width=\lwidth, red] (y3) -- ++(0.5*\xscale, 1*\yscale)  node[draw, circle, black, inner sep=\isize, fill=black!40] {};
        \draw[line width=\lwidth, red] (y3) -- ++(1.0*\xscale, .5*\yscale)  node[draw, circle, black, inner sep=\isize, fill=black!40] {};
        \draw[line width=\lwidth, red] (y4) -- ++(1*\xscale, 0.2*\yscale)  node[draw, circle, black, inner sep=\isize, fill=black!40] {};
        \draw[line width=\lwidth, red] (y2) -- ++(-0.7*\xscale, -1*\yscale)  node[draw, circle, black, inner sep=\isize, fill=black!40] {};
        \draw[line width=\lwidth, red] (y2) -- ++(0.4*\xscale, -1*\yscale)  node[draw, circle, black, inner sep=\isize, fill=black!40] {};
        \draw[line width=\lwidth, red] (y7) -- ++(0.9*\xscale,  0.4*\yscale)  node[draw, circle, black, inner sep=\isize, fill=black!40] {};
        \draw[line width=\lwidth, red] (y7) -- ++(0.4*\xscale, -0.8*\yscale)  node[draw, circle, black, inner sep=\isize, fill=black!40] {};
        \draw[line width=\lwidth, red] (y6) -- ++(0.8*\xscale, 0.7*\yscale)  node[draw, circle, black, inner sep=\isize, fill=black!40] {};
        \draw[line width=\lwidth, dash pattern={on \dashon pt off \dashoff pt}, blue] plot [smooth cycle, tension=.8] coordinates {(1.5*\xscale, -.2*\yscale) (3.5*\xscale, 1.2*\yscale) (4.6*\xscale, 1.9*\yscale) (5.8*\xscale, -1.5*\yscale)};
    \end{tikzpicture}            
    \end{center}
    \caption{Illustration of the connected components defined by an edge $e$. The cut weights is the maximal distance between a point in $L(e)$ and a point in $R(e)$.}
    \label{fig:cw}
\end{figure}
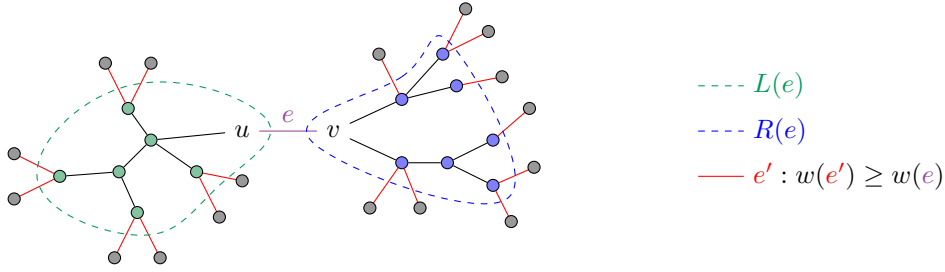
We say that $\ACW$ is an $\alpha$-approximation of the cut weights of $T$ if for every edge $e$ of $T$, we have
\[\cutw{e} \leq \acutw{e} \leq \alpha \cdot \cutw{e}.\]

Given a spanning tree $T$ and a edge-weight function $w: E_T \rightarrow \posreals$, the \emph{Cartesian tree} of $T$ with respect to $w$ is a weighted binary tree defined inductively as follows:
\begin{itemize}
    \item If $T$ is a single node, then the Cartesian tree is a single node with weight 0.
    \item Otherwise, the root of the Cartesian tree corresponds to the edge $e$ of $T$ with the largest weight, and the left and right children are the Cartesian trees of the two connected components of $T \setminus \{e\}$.
\end{itemize}

\citet{CKL20} provide an high-level algorithm that achieves a \(\gamma \cdot \alpha\)-approximation of $\BUF$ by computing a $\gamma$-KT and its $\alpha$-approximate cut weights. This algorithm is inspired by the one presented by \citet{FKW95}, who provide a quadratic time algorithm for computing the best ultrametric embedding of the data. The main difference is that in \cite{CKL20}, the first two steps are approximated rather than computed exactly. We outline this algorithm in \cref{algo:highest-level-approx}.

\begin{algorithm}[htbp]
    \caption{$\gamma \cdot \alpha$-approx. of the best ultrametric fit}
    \label{algo:highest-level-approx}
    \SetAlgoLined
    \KwIn{$X \subseteq \mathbb{R}^d$: set of points}
    \KwOut{A $\gamma \cdot \alpha$-approximation of \BUF}
    $T\gets \gamma$-KT $T$ of the complete graph induced by $X$\\
    $\ACW\gets \alpha$-approximation of the cut weights of  $T$\\
    $C_T \gets $ Cartesian tree of $T$ w.r.t. $\ACW$\\
    \Return{the ultrametric induced by $C_T$}
\end{algorithm}

\begin{theorem}[{\cite[Theorem 3.1]{CKL20}}]
    \label{thm:main-approx}
    For any $\gamma \geq 1$ and $\alpha \geq 1$, the output of Algorithm~\ref{algo:highest-level-approx} is a $\gamma \cdot \alpha$-approximation of the best ultrametric fit.
\end{theorem}

Furthermore step 3 of Algorithm~\ref{algo:highest-level-approx} can be easily computed in $O(n \log n)$ time using a disjoint-set data structure \citep{FKW95}. 
Hence, if there are algorithms that compute a $\gamma$-KT in time $f(n, d)$, and an $\alpha$-approximation of the cut weights in time $g(n, d)$, then there is an algorithm that computes a $\gamma \cdot \alpha$-approximation of the best ultrametric fit in time $\cO(n \log n) + f(n, d) + g(n, d)$.
In this work, we show that there are algorithms with $f(n,d) = \cOtilde(nd + n^{1+1/\gamma^2})$ and $g(n,d) = \cOtilde(nd + n^{1+1/\alpha^2})$. As mentioned in the preliminary section, we can assume w.l.o.g. $d = \cO(\log n)$, hence this term is asymptotically dominated by the other in both cases, and we drop it from now on.

\paragraph{Complexity of computing a $\gamma$-KT.}
\citet{CKL20} present an algorithm for computing a \(\gamma\)-KT in time \(\cOtilde\left(n^{1+O(1/\gamma^2)}\right)\).
This is achieved by constructing a sparse \(\gamma\)-spanner using the LSH-based algorithm of \citet{har2013euclidean} and then computing a minimum spanning tree of the \(\gamma\)-spanner.
However, as explained by \cite{andoni2023sub},
the \(O(1/\gamma^2)\) term in the exponent contains a factor of $12$, 
so this algorithm is only faster than the exact quadratic algorithm for \(\gamma > \sqrt{12}\).
Furthermore, \citet{andoni2023sub} prove that \(\gamma\)-spanners cannot\footnote{Unless extra nodes, called non-metric Steiner points, are added to the graph.} be constructed in subquadratic time for \(\gamma < \sqrt{2}\), thus making spanner-based approaches inefficient for building \(\gamma\)-KTs when \(\gamma < \sqrt{2}\). We address this limitation by introducing a new algorithm that computes a $\gamma$-KT in time \(\cOtilde\left(n^{1+1/\gamma^2}\right)\) without relying on spanners.

\paragraph{Complexity of approximating the cut weights.}
\begin{itemize}
    \item \citet{CKL20} provide an algorithm that computes a 5-approximation of the cut weights in time $O(nd + n \log n)$.
    \item \citet{CDL21} improve the previous result by using a data structure called coresets to provide an algorithm that computes a $(\sqrt{2} + \epsilon)$-approximation of the cut weights in time
    \[
    O\left( n \cdot d \cdot \left( \frac{\log \frac{1}{\varepsilon}}{\varepsilon^2} + \frac{\log n}{\varepsilon} \right) + n \cdot \frac{1}{\varepsilon^{4.5}} \log \frac{1}{\varepsilon} \right).
    \]
\end{itemize}

In both cases, the total asymptotic running time to compute the ultrametric is $\cOtilde(nd + n^{1+12/\gamma^2})$, dominated by the time complexity of the algorithm that computes a $\gamma$-KT (or more precisely, a $\gamma$-spanner in these algorithms). This observation indicates that there is room to improve the approximation factor of the cut weights without increasing the total running time of the algorithm. However, the methods presented in \cite{CKL20} and \cite{CDL21} cannot leverage this observation due to their inherent (geometric) bottlenecks, with approximation factors of 5 and $\sqrt{2}$ for computing the cut weights, respectively.

\paragraph{Our Contributions}

The main contributions of this paper are as follows: first, we restore the claim of \citet{CKL20} by providing an algorithm that, for any $\gamma \geq 1$, computes a $\gamma$-KT in time $\cOtilde(n^{1+1/\gamma^2} \log \delta)$ (\cref{thm:gammakt}), where $\delta$ is the \textit{spread} of the input space $X$, defined as the ratio between the diameter and the minimum pairwise distance of $X$. Second, we present an algorithm that computes an \(\alpha\)-approximation of the cut weights in time \(\cOtilde(n^{1+1/\alpha^2})\) for any \(\alpha \geq 1\) (\cref{thm:cw}).

\begin{restatable}{theorem}{gammakt}\label{thm:gammakt}
    For any $\gamma > 1$, there is an algorithm that computes a $\gamma$-KT of a given $n$-points Euclidean space $X$ in time and space $\cOtilde(n^{1+1/\gamma^2} \log \delta)$, where $\delta$ is the spread of $X$.
\end{restatable}

\begin{restatable}{theorem}{thmcw}\label{thm:cw}
    For any $\alpha > 1$, there is an algorithm 
    that computes a $\alpha$-approximation of the cut weights of an $n$-nodes tree $T$ in time and space $\cOtilde(n^{1+1/\alpha^2})$.
\end{restatable}

Finally, these two results can be combined to obtain a $c$-approximation of $\BUF$, as shown in the following corollary.

\begin{corollary}
    \label{cor:main}
    For any $c \geq 1$, there exists an algorithm that computes a $c$-approximation of $\BUF$ in time $\cOtilde(n^{1+1/c})$.
\end{corollary}
\begin{proof}
    Take $\gamma = \sqrt c$ and $\alpha = \sqrt c$ in \cref{thm:gammakt} and \cref{thm:cw}, respectively.
    By using \cref{thm:main-approx}, we obtain a $c$-approximation of $\BUF$, and the total running time is $\cOtilde(n^{1+1/c})$.
\end{proof}

For $c = (1 + \epsilon)$, this gives a $(1 + \epsilon)$-approximation of $\BUF$ in time $\cOtilde(n^{2 - \epsilon + o(\epsilon^2)})$, which remains subquadratic in $n$.
For comparison, when $c=1$ (no approximation), it is known that the best ultrametric fit problem cannot be solved in subquadratic time, see for example \cite{FKW95}.

\paragraph{High-level overview of the techniques.}
\begin{itemize}
    \item \textbf{Theorem~\ref{thm:gammakt}:} We achieve this result using a new method that avoids \(\gamma\)-spanner constructions. The algorithm operates essentially through a breadth-first traversal of the graph, guided by locality-sensitive hashing of the data.
    \item \textbf{Theorem~\ref{thm:cw}:} This algorithm is built on a dynamic version of the approximate farthest neighbor data structure developed by \citet{pagh2017approximate}, which supports queries in time \(\cOtilde(n^{1/\alpha^2})\) and space \(\cOtilde(n^{1+1/\alpha^2})\). We extend this data structure to work over a partition of the input space \(X\), allowing approximate farthest neighbor queries within \textit{clusters} (i.e., subsets of the partition) and efficient merging of clusters.
\end{itemize}

We now proceed to the details of these two algorithms.

\section{Algorithm for $\gamma$-Kruskal tree}\label{sec:kt}
In this section, we present the algorithm to compute a \(\gamma\)-KT, as stated in \cref{thm:gammakt}, along with its full analysis.

Recall that the algorithm by \citet{CKL20} for computing a \(\gamma\)-KT operates in two steps: first, it computes a sparse \(\gamma\)-spanner of the complete graph induced by the set of points \(X\), and then it computes a minimum spanning tree of the \(\gamma\)-spanner. To find a \(\gamma\)-spanner, they use the algorithm of \citet{har2013euclidean}, which runs in time \(\tilde{\mathcal{O}}(n^{1+12/\gamma^2})\) and uses Locality-Sensitive Hashing.
This algorithm is subquadratic only when \(\gamma < \sqrt{12}\). While it might be possible to reduce this constant, \citet{andoni2023sub} showed that spanners cannot be constructed in subquadratic time for \(\gamma < \sqrt{2}\).

Instead of relying on spanners, we propose an algorithm that builds a \(\gamma\)-KT for any \(\gamma \geq 1\) in time \(\cOtilde(n^{1+1/\gamma^2} \log \delta)\), essentially via a breadth-first traversal guided by LSH.

We first recall the definition of Locality-Sensitive Hash functions (LSH):
\begin{definition}
    Let $(X, \metric)$ be a metric space over $n$ points, and let $c > 1$ and $R > 0$.
    A family $\Hh$ of functions is a $(\rho, c, R)$-LSH if,
    when choosing a function $h$ uniformly at random from $\Hh$
    we have, for every $x,y \in X$:
    \begin{itemize}
        \item $\metric(x, y) \le R \Rightarrow \Prob{h(x) = h(y)} \ge 1/n^{\rho}$
        \item $\metric(x, y) \ge cR \Rightarrow \Prob{h(x) = h(y)} \le 1/n$
    \end{itemize}
\end{definition}

\citet{andoni2006near} showed that when $X$ is a Euclidean space, i.e.
$\metric$ is the $\euclideanmetric$ metric, then for every $c > 1$ and every $R >0$,
there exists a $(1/c^2, c , R)$-LSH family of functions $\Hh_{c,R}$,
and the evaluation of a random function from $\Hh_{c,R}$ on all $n$ points of $X$ takes $\cOtilde(n)$ time.

Given an LSH function $h$, the \textit{buckets} of $h$ are the equivalence classes of the relation $x \sim y \Leftrightarrow h(x) = h(y)$. Intuitively, points with the same hash value are put into the same bucket.

Finally, we introduce an important property, denoted (*), that will be useful in explaining the behavior of our algorithm.
\begin{definition}
    We say that a set $E$ of edges satisfies the property (*) if for any pair $(u,v) \in X^2$, there is a path from $u$ to $v$ in $E$ using only edges of weight at most $\gamma \cdot \euclideanmetric(u,v)$.
\end{definition}

\paragraph{Overview of the algorithm.} The algorithm works in two steps:
\begin{itemize}
    \item First we compute a set $E$ of $\cOtilde(n^{1+1/\gamma^2}\cdot \log \delta)$ edges that satisfies (*). This is the purpose of \cref{prop:algo-property-star}.
    \item Then, we compute a minimum spanning tree of $E$. \cref{lemma:aux:kt} shows that this tree is a \(\gamma\)-KT. This part is done in time $\cO(|E|\log |E|) = \cOtilde(n^{1+1/\gamma^2})$ using the standard algorithm of \citet{kruskal1956shortest}.
\end{itemize}

The above algorithm, combined with the following proposition and lemma, yields the main result of this section, a proof of \cref{thm:gammakt}.

\begin{restatable}{proposition}{algokt}\label{prop:algo-property-star}
    There is an algorithm that computes a set $E$ of $\cOtilde(n^{1+1/\gamma^2}\cdot \log \delta)$ edges that satisfies (*) in time $\cOtilde(n^{1+1/\gamma^2} \log \delta)$.
\end{restatable}

\begin{restatable}{lemma}{lemmaauxkt}\label{lemma:aux:kt}
    Let $X$ be an $n$-points Euclidean space and let $\gamma > 1$.
    Let $E \subseteq X^2$ be a set of edges that satisfies (*).
    then running Kruskal's algorithm on $E'$ yields a $\gamma$-KT of $X$ in time $\cO(|E| \log |E|)$. 
\end{restatable}

\subsection{Proof of \cref{prop:algo-property-star}}
To prove \cref{prop:algo-property-star}, we introduce \cref{algo:kt-bfs} that will be used as a useful subroutine in the proof.

\begin{algorithm}[htbp]
    \caption{Local-BFS in a graph}\label{algo:kt-bfs}
    \SetAlgoLined
    \KwIn{$X \subseteq \mathbb{R}^d$: set of points,\\
    $\gamma$: approximation parameter,\\
    $R$: target radius
    }
    \BlankLine
    $E \gets \emptyset$\;
    $h \gets (1/\gamma^2, \gamma, R)\textsc{-LSH}(X)$\;
    \ForEach{bucket $B$ of $h$}{
        $\texttt{coll} \gets 0$\;
        $S \gets B$\;
        \While{$S$ is not empty}{
            $x \gets S.\texttt{pop\_any}()$\;
            $q \gets \textsc{Queue}()$; $q.push(x)$\;
            \While{$q$ is not empty}{
                $u \gets q.pop()$\;
                \ForEach{\label{line:progress}$v$ in $S$}{
                    \eIf{$\euclideanmetric(u,v) \le \gamma\cdot R$}{
                        $S.\texttt{remove}(v)$\;
                        $E \gets E\cup \{(u,v)\}$\;\label{line:add_edge}
                        $q.push(v)$\;
                    }{
                        $\texttt{coll} \gets \texttt{coll} + 1$\;
                    }
                }
            }
        }
    }
    \Return{E}
\end{algorithm}

We start by proving properties of \cref{algo:kt-bfs}.
We first show that it runs in expected quasilinear time.
\begin{lemma}\label{lemma:algo-kt-bfs-complexity}
    \cref{algo:kt-bfs} runs in $\cOtilde(n)$ expected time.
\end{lemma}
\begin{proof}
    As mentioned above, applying the LSH function $h$ to all points of $X$ takes $\cOtilde(n)$ time.

    We show that for each bucket $B$ of $h$, the body of the outer \texttt{foreach} loop
    runs in time $\cO(|B| + collisions(B))$, where $collisions(B)$ is the value
    of \texttt{coll} at the end of the \texttt{foreach} iteration on bucket $B$.
    
    By implementing the removal from $S$ as a filtering procedure, both branches of the \texttt{if}
    statement take constant time.
    
    Next, the algorithm enters the body of the \texttt{foreach} loop on \cref{line:progress}
    at most $|B| + collisions(B)$ times,
    as the loop body either removes a vertex for $S$, which has size $|B|$,
    or increases the value of \texttt{coll}.
    Since each vertex is pushed (and therefore popped) from the queue at most once,
    the outer \texttt{while} loop runs in time  $\cO(|B| + collisions(B))$,
    and the \cref{algo:kt-bfs} runs in time $\cOtilde(n + \sum_B collisions(B))$.

    It remains to show that the expected value of $\sum_B collisions(B)$ is $\cO(n)$.
    The variable \texttt{coll} is incremented whenever we encounter a pair of vertices $u,v$ that fall into the same bucket and $\euclideanmetric(u,v) > \gamma R$.
    By definition of the LSH function, this event occurs with probability at most $1/n$, hence
    over all pairs $u,v$ we have at most $n$ such collisions in total in expectation,
    and therefore \cref{algo:kt-bfs} runs in time $\cOtilde(n)$.
\end{proof}

We now show that the set $E$ of edges computed by \cref{algo:kt-bfs} has a desired connectivity property.
\begin{lemma}\label{lemma:algo-kt-bfs-correct}
    \cref{algo:kt-bfs} returns a set $E$ of up to $n$ edges of weight at most $\gamma R$
    such that for any pair of points $(u,v) \in X^2$
    with $\euclideanmetric(u,v) \le R$ and $h(u) = h(v)$,
    there is a path from $u$ to $v$ in $E$.
\end{lemma}
\begin{proof}
    Since we only add to $E$ edges of weight at most $\gamma R$, the condition on edge weights is satisfied.
    Furthermore, when adding an edge $(u,v)$ to $E$, we remove $v$ from $S$ and it is never added back, hence $E$ contains at most $n$ edges.

    To prove the connectivity property,
    we show that the body of the outer \texttt{while} loop adds to $E$
    a set of edges that spans the subset of vertices $u$ in $B$ that are connected to $x$
    with edges from of weight at most $\gamma R$.
    Let $G_{\gamma R}$ denote the graph with vertex set $X$ and all edges of length at most $\gamma R$.
    First, notice that the inner \texttt{while} loop of \cref{algo:kt-bfs} performs a breadth-first traversal,
    i.e. the vertices are enumerated in order of nondecreasing distance $d_G$ to $x$,
    where $d_G(u,v)$ defined as the minimum number of edges in a $u$--$v$ path in $G_{\gamma R}$, and $+\infty$ is there is no such path.

    We show by induction on $d_G(x, u)$ that for every $u\in B$ such that $d_G(x,u)$ is finite, then there is a path from $x$ to $u$ in $E$.
    The base case is when $d_G(x, u) = 0$, i.e. $x = u$: the property is trivially satisfied.
    Next, let $u\in B$ be such that $d_G(x, u)= k+1$.
    By definition, every neighbor $v$ of $u$ is at distance at least $k$ of $x$, and at least one of them satisfies $d(x,v) = k$.
    Consider the first such $v$ enumerated by the algorithm: all other neighbors of $u$ will be enumerated after $v$, therefore, at this point, $u$ has not been removed from $S$, hence we add the edge $(u,v)$ to $E$.
    By induction hypothesis, there is a path from $x$ to $v$ in $E$, and the edge $(u,v)$ ensures that there is also a path from $x$ to $u$ in $E$, concluding our induction.
\end{proof}

Next, we use \cref{algo:kt-bfs} to build a set of edges that connects every pair of points at distance at most $R$ with edges of weight at most $\gamma R$.
\begin{corollary}\label{coro:kt-bfs}
    Let $X$ be a subset of $n$ points in $\reals^d$, let $\gamma > 1$, and let $R$ be a target radius.
    There is an algorithm that runs in expected time $\cOtilde(n^{1 + 1/\gamma^2})$
    and, with probability at least $1 - 1/n$,
    returns a set $E \subseteq X^2$ of at most $\cOtilde(n^{1 + 1/\gamma^2})$ edges,
    each of length at most $\gamma R$
    such that for any pair $(u,v)\in X^2$ of distance at most $R$,
    there is a path from $u$ to $v$ in $E$.
\end{corollary}
\begin{proof}
    We run \cref{algo:kt-bfs} $L = 3 n^{1/\gamma^2} \log n$ times independently with parameters $(\gamma, R)$, and return the union of all resulting sets $E$.

    Consider a pair $u,v$ of points of $X$ such that $\euclideanmetric(u,v) \le R$.
    If at some step $i=1,\ldots, L$ of the above iteration, $u$ and $v$ are in the same bucket of $h$, then by \cref{lemma:algo-kt-bfs-correct}, there will be a path from $u$ to $v$ in $E$.
    As $\euclideanmetric(u,v) \le R$, the probability that $u$ and $v$ fall in the same bucket at a fixed step $i$ is at least $n^{-\rho}$.
    Therefore, using independence of the runs, the probability that $u$ and $v$ never fall in the same bucket over all $L$ runs is at most 
    \[(1-n^{-\rho})^L \le e^{-3\log n} = 1/n^3.\]
    Furthermore, the algorithm fails only when there exists a pair $u,v$ such that $\euclideanmetric(u,v) \le R$
    that never fall in the same bucket.
    By union bound over all pairs $u,v$, this event occurs with probability at most $1/n$.

    The running time of this procedure follows from \cref{lemma:algo-kt-bfs-complexity}.
\end{proof}

\begin{proof}[Proof of \cref{prop:algo-property-star}]
    By Corollary~\ref{coro:kt-bfs}, there is an algorithm which, when used with a fixed value $R$, ensures that any two points $u, v$ such that $\euclideanmetric(u,v) \le R$ are connected by a path of edges of length at most $\gamma \cdot R$ with probability at least $1-1/n$.
    
    To create the set $E$ satisfying property (*), we run this algorithm for $O(\log n \cdot \log \delta)$ values of $R$, in order to cover the range of distances between points in $X$. 
    More precisely let $d_{\min}$ and $d_{\max}$ denote the minimum and maximum distance between distinct points in $X$, and let $\delta = d_{\max}/d_{\min}$ denote the spread of $X$.
    Let  $\tau = 1+1/\log n$. By using the algorithm of \cref{coro:kt-bfs} with parameter $\gamma' = \gamma / \tau$ and $R = d_{\min}, d_{\min}\tau, d_{\min}\tau^2, \ldots, \tau d_{\max}$, we get a set $E$ that connects any pair of points at distance $d$ with edges of
    weight $\tau \cdot (\gamma/\tau) \cdot d = \gamma \cdot d$.
    The number of calls to \cref{coro:kt-bfs} is, up to a constant 
    \[\log_{\tau}(d_{\max}/d_{\min}) = \frac{\log(\delta)}{\log \tau} = \cO(\log n \cdot \log (\delta)).\]

    The running time of the algorithm of \cref{coro:kt-bfs} with parameter $\gamma'$ is $\cOtilde(n^{1 + 1/\gamma'^2})$, which is $\cOtilde(n^{1 + 1/\gamma^2})$: this concludes the proof.
\end{proof}

\subsection{Proof of \cref{lemma:aux:kt}}
We are now ready to prove \cref{lemma:aux:kt}.

\begin{proof}[Proof of \cref{lemma:aux:kt}]
    Kruskal's algorithm builds a minimum spanning tree by first sorting the edges in $E$ in order of non-decreasing weights,
    and then iterating over all edges in order, adding to $T$ each edge that connects two disjoint connected components of $T$.

    Assume that for the sake of contradiction that the output tree $T$ is not a $\gamma$-KT.
    Then there exists an edge $(u,v)$ of weight $w = \euclideanmetric(u,v)$ such that there is an edge $e'$ of weight $\euclideanmetric(e') > \gamma\cdot w$
    on the path from $u$ to $v$ in $T$.
    Removing $e'$ from $T$ yields two connected components $C_1$ and $C_2$, with $u\in C_1$ and $v\in C_2$. 
    By assumption, there is a $u-v$ path consisting of edges in $E'$, each of weight at most $\gamma\cdot w$.
    As $u\in C_1$ and $v\in C_2$, one edge $e^*$ of this path is not in $T$ and has one endpoint in $C_1$ and the other in $C_2$.
    As $\euclideanmetric(e^*) < \euclideanmetric(e')$, $e^*$ was considered before $e'$ by the algorithm, i.e. at a time where $C_1$ and $C_2$ were disjoint. Therefore, the algorithm has added $e^*$ to $T$, a contradiction as $e^* \notin T$.
\end{proof}

\section{Better Cut-Weights via Approximate Farthest Neighbors} \label{sec:cw}

The second step of \cref{algo:highest-level-approx} is to compute the cut weights of the spanning tree obtained in the first step. \citet{FKW95} provided a quadratic-time algorithm to compute the exact cut weights. More recently, \cite{CKL20} and \cite{CDL21} proposed a $5$- and a $\sqrt{2}$-approximation algorithm that both operate in quasilinear time. However, their approximations are inherently limited due to their reliance on specific geometric properties of Euclidean spaces.

In this section, we introduce an $\alpha$-approximation algorithm for the cut weights that works for any \(\alpha > 1\). Our result is as follows:

\thmcw*

The core ingredient of our algorithm is a \textit{dynamic} version of the data structure for approximate farthest neighbor of \citet{pagh2017approximate}, which we present in the next subsection.
The algorithm behind the \cref{thm:cw} is described in the second part of this section.

\subsection{Dynamic Approximate Farthest Neighbor}\label{sec:afn}

In order to obtain an $\alpha$-approximation of the cut weights, we use the data structure of \citet{pagh2017approximate} for \textit{approximate farthest neighbors} (AFN) in Euclidean spaces. Their data structure preprocesses a subset of a metric space,
and can then find in that subset an $\alpha$-approximate farthest neighbor of a given query point in time $\cOtilde(n^{1/\alpha^2})$.

\begin{definition}[\afnpb]\label{def:afn}
    Let $(X, d)$ be a metric space and let $\alpha > 1$.
    A data structure $\Dd$ solves the $\alpha$-\afnpb problem over a set $S\subseteq X$ if,
    given a point $q\in X$, it returns a point $r$ that is an $\alpha$-approximate farthest neighbor of $q$ in $S$, i.e. it holds that
    \[d(q, r) \ge \frac{1}{\alpha} \max_{p\in S} d(q, p).\]
\end{definition}

To fit our use case, we show that one can extend the data structure of \citet{pagh2017approximate} to be \textit{dynamic}, i.e. given the data structure for two disjoint subsets $S, S'$ of $(X, d)$, we can construct a data structure for $S \sqcup S'$ faster than the time needed to build it from scratch.

\begin{restatable}{theorem}{thmafn}\label{thm:afn}
    There is a data structure for $\alpha$-\afnpb over a dynamic partition of a metric space $(X, d)$ of $n$ points that supports the following operations:
    \begin{itemize}
        \item \textsc{Initialize}$(X, \alpha)$: create an data structure containing a cluster $S_x = \{x\}$ for each $x$ in $X$, in time $\cOtilde(n^{1+1/\alpha^2})$.
        \item \textsc{Query}$(\Dd, S, q)$: given a cluster $S$ in $\Dd$ and a query point $q\in X$, return an $\alpha$-approximate farthest neighbor of $q$ in $S$ in time $\cOtilde(n^{1/\alpha^2})$.
        \item \textsc{Merge}$(\Dd, S, S')$: given two clusters $S, S'$ in $\Dd$, add the cluster $S'' = S \sqcup S'$ to $\Dd$ in time $\cOtilde(n^{1/\alpha^2} \cdot \min(|S|, |S'|, n^{1/\alpha^2}))$.
        This operation consumes the clusters $S$ and $S'$, i.e. they cannot be used in other operations afterward.
    \end{itemize}
    This data structure uses $\cOtilde(n^{1+1/\alpha^2})$ space.
    Furthermore, this construction is probabilistic (in the \textsc{Initialize} function), and for every $q$ and $S$, the $\textsc{Query}$ operation fails with probability at most $1/n^3$.
\end{restatable}

Before proving \cref{thm:afn}, we briefly recall how the approximate farthest neighbor data structure of \citet{pagh2017approximate} works.
We will then show how to adapt it to our needs.

Let $p_1, p_2, \dots, p_n$ be the points in $S$. Intuitively, the data structure of Pagh et al. uses projections on many random lines to identify points in the input set that are ``extremal'' along some direction, and search the farthest point of a given query point in this subset of extremal points.
More precisely, the data structure samples $L = \cO(n^{1/\alpha^2})$ Gaussian random vectors $a_1,\ldots, a_L$.
Then, for each $i \le L$ and each $j = 1,\ldots, n$, let $\beta_{ij}$ denote the value of the inner product $\langle a_i, p_j\rangle$.
For each $i$, the structure stores a collection $\Cc_i$ containing the (up to) $M = \cOtilde(n^{1/\alpha^2})$ pairs $(j, \beta_{ij})$
that have the largest value of $\beta_{ij}$. 
Storing these indices requires $\cO(n^{1/\alpha^2} \cdot \min(|S|, n^{1/\alpha^2}))$ space.

While this data structure may store up to $\cOtilde(n^{2/\alpha^2}$ candidate points, Pagh et al. show how to select a subset $S'$ of $M$ points that depends on the query point $q$ so that with constant probability, $S'$ contains an $\alpha$-approximate farthest neighbor of $q$.
Under the assumption that for every $i$, there is an efficient way of iterating over the $p_j$ in order of decreasing values of $\langle a_i, p_j\rangle$, their query algorithm runs in time $\cOtilde(n^{1/\alpha^2})$.

Note that the randomness for the probability of error is taken over the choice of the $a_i$'s at construction not: once these are fixed, the query algorithm is deterministic.
By taking $\Theta(\log n)$ independent copies of this data structure and returning the farthest point from $q$ across all queries, we can reduce the probability of error to less than $1/n^3$.

We are now ready to explain how to build the dynamic data structure of \cref{thm:afn}.
\begin{proof}[{Proof of \cref{thm:afn}}]
    We first select $L' = L \cdot \Theta(\log n) = \cOtilde(n^{1/\alpha^2})$ random Gaussian vectors $a_1,\ldots, a_{L'}$.
    The data structure $\Dd$ will maintain a dynamic partition of $X$, and for each \textit{cluster} (i.e. each set) $S$ in the partition,
    store for every $i$ the $\min(|S|, M)$ points that maximize the value of $\beta_{ij} =\langle a_i, p_j\rangle$.

    For the \textsc{Initialize} operation, each cluster contains a single point: we only need to compute the values $\beta_{ij}$ for every $i$
    and $j$, which takes time $\cOtilde(n^{1+1/\alpha^2})$ as $d = \cO(\log n)$.

    For every $i$, we store the collection $\Cc_i$ of pairs $(j, \beta_{ij})_{j = 1,\ldots, M}$ in a binary search tree using $\beta_{ij}$ as key. As we can efficiently iterate over the elements of a binary search tree in order of decreasing keys,
    the \textsc{Query} operation runs in time $\cOtilde(n^{1/\alpha^2})$ using the algorithm of Pagh et al.
    Furthermore, as $L' = L \cdot \Theta(\log n)$, the probability of returning a point that is not a $c$-approximate neighbor of the query point $q$ is at most $1/n^3$. 

    Finally, using binary search trees allows us to implement the \textsc{Merge} operation in time $\cOtilde(n^{1/\alpha^2} \cdot \min(|S|, |S'|, n^{1/\alpha^2}))$. Let $S$ and $S'$ be two clusters, and assume w.l.o.g. that $|S| \le |S'|$.
    To merge $S$ and $S'$, we move all elements from the cluster $\Cc_i$ of $S$ to the corresponding cluster $\Cc_i'$ of $S'$, and then truncate $\Cc_i'$ to keep only the $M$ pairs with largest key.
    As $|\Cc_i'|\le n$, inserting an element into $\Cc_i'$ takes time $\cO(\log n)$ and moving the elements for a fixed $i$ takes time $\cO(|\Cc_i| \cdot \log n)$.
    As $\Cc_i$ contains $\cO(\min(|S|, n^{1/\alpha^2})$ elements and there are $L'  =\cOtilde(n^{1/\alpha^2})$ collections $\Cc_i$ to move,
    this takes a total of $\cOtilde(n^{1/\alpha^2} \cdot \min(|S|, |S'|, n^{1/\alpha^2}))$ time.

    Moving elements instead of copying them ensures that the space total usage remains $\cO(n^{1+1/\alpha^2})$.
\end{proof}

\subsection{Approximate Cut-weight Algorithm}\label{sec:cw-algo}

We give a proof of \cref{thm:cw}, i.e. we give an algorithm that computes an $\alpha$-approximation of cut weights of a tree in time and space $\cOtilde(n^{1+1/\alpha^2})$, using the data structure of \cref{thm:afn}.

By augmenting the aforementioned data structure with a disjoint-set data structure,
we can additionally support the following operations:
\begin{itemize}
    \item \textsc{Find}$(\Dd, q)$: return the (unique) cluster $S$ in $\Dd$ that contains $q$, in time $O(\log^* n)$.
    \item \textsc{Enumerate}$(\Dd, S)$: iterate over all elements of a cluster $S$ in $\Dd$, in total time $O(|S|)$.
\end{itemize}
The procedure to compute an $\alpha$-approximation of the cut weights is given in \cref{algo:cw}.

\begin{algorithm}[htbp]
    \caption{$\alpha$-approximation of the cut weights}\label{algo:cw}
    \SetAlgoLined
    \KwIn{$X \subseteq \mathbb{R}^d$: set of points,\\
    $T$: spanning tree as a list of weighted edges sorted by non-decreasing weight,\\
    $\alpha > 1$: approximation parameter}
    \BlankLine
    $\Dd \gets \textsc{Initialize(X)}$\;
    \ForEach{\label{line:foreach}edge $e = (x, y)$ in increasing order of weights}{
        $S_x \gets \textsc{Find}(\Dd, x)$;\label{line:find}
        $S_y \gets \textsc{Find}(\Dd, y)$;

        \If{$|S_x| > |S_y|$}{
            Swap $S_x$ and $S_y$\;\label{line:swap}
        }
        $\ACW(e) \gets \alpha \cdot \max \{ w(z, \textsc{Query}(\Dd, S_y, z)) | z \in \textsc{Enumerate}(\Dd, S_x) \}$\;\label{line:acw}
        $\textsc{Merge}(\Dd, S_x, S_y)$\;\label{line:merge}
    }
    \Return{\ACW}
\end{algorithm}

We now turn to proving \cref{thm:cw}.
We first analyze the complexity of \cref{algo:cw}.
The space complexity of $\cOtilde(n^{1+1/\alpha^2})$ follows from that of the data structure $\Dd$ of \cref{thm:afn}.
To obtain the desired time complexity, we crucially rely on the fact that, on \cref{line:acw} of \cref{algo:cw},
we iterate over the smallest of $S_x$ and $S_y$ and query the other when computing $\ACW$.
This allows us to prove using a counting argument that the algorithm makes $\cO(n\log n)$ calls to \textsc{Query},
showing that the algorithm runs in time $\cOtilde(n^{1+1/\alpha^2})$.
\begin{restatable}{lemma}{lemmaalgocwcomplexity}\label{lemma:algocw-complexity}
    \cref{algo:cw} runs in time $\cOtilde(n^{1+1/\alpha^2})$.
\end{restatable}
\begin{proof}
    First, we show that the body of the \texttt{foreach} loop on \cref{line:foreach} of \cref{algo:cw}
    runs in time $\cOtilde(n^{1/\alpha^2} \cdot \min(|S_x|, |S_y|)$.
    The two \textsc{Find} operations on \cref{line:find} run in time $O(\log^* n) = \cOtilde(n^{1/\alpha^2})$.
    Then, after \cref{line:swap}, $|S_x|$ is at most $|S_y|$, hence we show that the rest of the loop body runs in time
    $\cOtilde(n^{1/\alpha^2} \cdot |S_x|)$.
    On \cref{line:acw}, we make a single call to \textsc{Enumerate}, plus one distance computation and one call to \textsc{Query} for each $z\in S_x$ 
    which takes a total time of $\cOtilde(n^{1/\alpha^2} \cdot |S_x|)$.
    Finally, from \cref{thm:afn}, the call to \textsc{Merge} on \cref{line:merge} takes time $\cOtilde(n^{1/\alpha^2} \cdot \min(|S_x|, |S_y|, n^{1/\alpha^2}))$, which is also dominated by $\cOtilde(n^{1/\alpha^2} \cdot |S_x|)$.

    Now, we show that the sum of $\min(|S_x|, |S_y|)$ over the course of the algorithm is at most $\cO(n \log n)$.
    At the start of the algorithm, each point $z\in X$ belongs to a cluster of size $1$, and, as $T$ is a spanning tree, they will all be in a single cluster of size $n$ when the procedure ends.
    If during an iteration of the loop, $z$ was in the smaller cluster, the size of its cluster will increase at least twofold.
    Therefore, this can happen at most $\cO(\log n)$ times to each point.
    Hence, over all points in $X$, the event ``being in the smallest cluster'' happens at most $\cO(n \log n)$ times.

    Therefore, the running time of the algorithm is $\cOtilde(n^{1+1/\alpha^2})$.
\end{proof}

We now show that the $\ACW$ function is an $\alpha$-approximation of the cut weights with high probability.
Intuitively, putting aside low-probability errors, $\textsc{Query}(\Dd, S_y, q)$ returns a point whose distance
to $q$ is between $d_{\max}/\alpha$ and $d_{\max}$, where $d_{\max}$ is the maximum distance between $q$ and a point of $S_y$.
Therefore, by taking the maximum over all points in $S_x$ and multiplying by $\alpha$,
we obtain a value between $\CutW(e)$ and $\alpha\cdot\CutW(e)$. 
\begin{restatable}{lemma}{lemmaalgocwcorrect}\label{lemma:algocw-correct}
    With high probability, the function $\ACW$ returned by \cref{algo:cw} is an $\alpha$-approximation of
    the cut weights $\CutW$ of $T$.
    More precisely, with probability at least $1 - 1/n$, we have, for every $e\in T$:
    \[\CutW(e) \le \ACW(e) \le \alpha \cdot \CutW(e).\]
\end{restatable}

\begin{proof}
    First, assume that no call to \textsc{Query} returns an incorrect result: we show that in this case, $\ACW$ is an $\alpha$-approximation of the cut weights, and will then show that this event occurs with probability at least $1-1/n$.
    One can show by induction that, at each iteration, we have $S_x = L(e)$ and $S_y = R(e)$ where $e = (x,y)$.
    Therefore, we can rewrite the cut weight of $e$ as
    \[\CutW(e) = \max_{z\in S_x} \max_{z'\in S_y} w(z, z').\]
    Furthermore, the $\textsc{Query}(\Dd, S_y, z)$ algorithm returns a point $r_z$ such that
    $\alpha \cdot w(z, r_z) \ge \max_{z'\in S_y} w(z, z')$.
    Hence, we have:
    \begin{align*}
    \ACW(e) &= \alpha \cdot \max_{z\in S_x} w(z, r_z)\\
                        &\ge \max_{z\in S_x} \max_{z'\in S_y} w(z, z')\\
                        &= \CutW(e)
    \end{align*}
    On the other hand, $w(z, z_r)$ cannot exceed $\CutW(e)$, hence we have $\CutW(e) \le \ACW(e) \le \alpha \CutW(e)$:
    this shows that $\ACW$ is an $\alpha$-approximation of $\CutW$.

    The algorithm makes $\cO(n \log n) = o(n^2)$ calls to \textsc{Query} (see the second part of the proof of \cref{lemma:algocw-complexity})
    therefore by union bound, it errs with probability $o(1/n)$, which is less than $1/n$ for large enough $n$.
\end{proof}

\section{Experiments}\label{sec:experiments}

We evaluate the performance of our algorithm on two types of datasets. First, as in~\cite{CKL20} and \cite{CDL21}, we use five classic real-world datasets to evaluate both the quality of the approximation and the runtime of the algorithms; see \cref{tab:datasets} for details on these datasets.
Second, we use synthetic datasets to evaluate how the runtime of the algorithms scales with larger datasets. In this case, we cannot measure the quality of the approximation given by the algorithms, since this takes quadratic time, which is unreasonably long for large datasets.
We compare our algorithm with the state-of-the-art algorithm of \citet{CDL21}
and the widely used implementation of the \texttt{fastcluster} Python package.

\begin{table}[htbp]
    \centering
    \begin{tabular}{l|rcc}
    \toprule
    \textbf{Dataset} & \multicolumn{1}{|c}{\textbf{size}} & \textbf{dim.} & \textbf{optimal dist.} \\
    \midrule
    \iris & 150 & 4 & 8.07 \\
    \diabetes & 768 & 8 & 5.96 \\
    \mice & 1080 & 77 & 4.92 \\
    \pendigits & 10992 & 16 & 13.86 \\
    \shuttle & 58000 & 9 & 29.72 \\
    \bottomrule
    \end{tabular}
    \caption{Description of the datasets used for evaluation.
    All datasets are publicly available on the UCI ML Repository~\citep{uci_repo}, or Kaggle~\citep{kagglediabetes} for the \diabetes{} dataset.}
    \label{tab:datasets}
    \end{table}

\begin{table*}[htbp]
    \centering
    \begin{tabular}{l|cacaca}
        \toprule
        &\multicolumn{2}{c}{MICE}&\multicolumn{2}{c}{PENDIGITS}&\multicolumn{2}{c}{SHUTTLE}\\
        \multicolumn{1}{c|}{Algorithm} & apx. & $T~(s)$& apx. & $T~(s)$& apx. & $T~(s)$\\
        \midrule
         FKW& $1$ & $0.12s$& $1$ & $8.16s$& $1$ & $236.72s$\\
         \fastbuf($c = 4)$& $1.88$ & $0.12s$& $1.53$ & $1.80s$& $1.41$ & $21.39s$\\
         \fastbuf($c = 9)$& $3.67$ & $0.07s$& $2.26$ & $0.79s$& $1.85$ & $7.82s$\\
         \fastbuf($c = 16)$& $5.58$ & $0.04s$& $2.79$ & $0.52s$& $2.63$ & $4.69s$\\
         CVL& 3.27 & $0.14s$& 1.85 & $0.76s$& 2.41 & $10.17s$\\
        \bottomrule
    \end{tabular}
    \caption{Comparison of \fastbuf{} with the state-of-the-art algorithm by \citet{CDL21}, denoted "CVL" in the table. The ``apx.'' column reports the approximation factor, i.e. the distortion of the output ultrametric normalized by the distortion of the optimal ultrametric, given by quadratic-time algorithm by \citet{FKW95}, denoted "FKW". 
    Each reported distortion or running time value is the average of 30 runs of the algorithm, all standard deviations are less than 10\% for appoximation and $2\%$ for runtimes.}
    \label{table:cmp}
\end{table*}

Experiments presented in this paper were carried out using the Grid'5000 testbed, supported by a scientific interest group hosted by Inria and including CNRS, RENATER and several Universities as well as other organizations (see \url{https://www.grid5000.fr}). 
The experiments were conducted on identical nodes of the Grid'5000 cluster, running Debian GNU/Linux 5.10.0-28-amd64. The hardware configuration includes an Intel(R) Xeon(R) CPU E5-2630 v3 @ 2.40GHz and 126GB of RAM. For the experiments, our algorithm is implemented in the Rust programming language, version 1.79.0 (\texttt{129f3b996}, 2024-06-10). Our code was compiled in \texttt{release} mode.

\subsection{Main experiments}

\paragraph{Experiment 1: Accuracy of the \BUF{} algorithm.}

We measure the performance of our \BUF{} $c$-approximation algorithm using our $\gamma$-KT and $\alpha$-ACW algorithms with $\alpha = \gamma = \sqrt{c}$, which we call \fastbuf.
Based on the results of Experiment A (see Paragraph~\ref{par:experiment_A}), we modify the $\alpha$-ACW algorithm to multiply distances by $\sqrt{\alpha}$ instead of $\alpha$: in practice, this is sufficient to overestimate the cut weights and improves the approximation factor. 
We evaluate \fastbuf for different values of $c$, and for each value we
run the algorithm $t = 30$ times on each of the $5$ datasets, for a total of $150$ runs per value of $c$.

The key takeaway from this experiment is that our algorithm performs significantly better than the worst-case approximation factor $c$. For example, if one wants a $2$-approximation of the best ultrametric embedding of a dataset, they can run the algorithm with a larger parameter, e.g. $c = 9$ instead of $2$. This approach will greatly reduce the running time and space usage from $\tilde{\mathcal{O}}(n^{1.5})$ to $\tilde{\mathcal{O}}(n^{1.11})$ while still providing a high-quality, low-distortion embedding.

We ran \fastbuf for $c = 2, 4, 9, 16, 100, 400$: of the 900 resulting runs, \fastbuf failed to over-approximate the cut weights in only $20$ runs, which is less than $2.3\%$ of the runs.
In addition, the algorithm remains very accurate. For example, when using $c = 9$, the procedure returned $2$-approximation of the best ultrametric fit in all but 5 of the 150 runs. The results for $c = 2, 4, 9, 16$ are reported in \cref{fig:buf_var} (values $c = 100, 400$ are omitted here for readability of the figure).
\begin{figure}[h]
    \centering
    \includegraphics[width=0.7\textwidth]{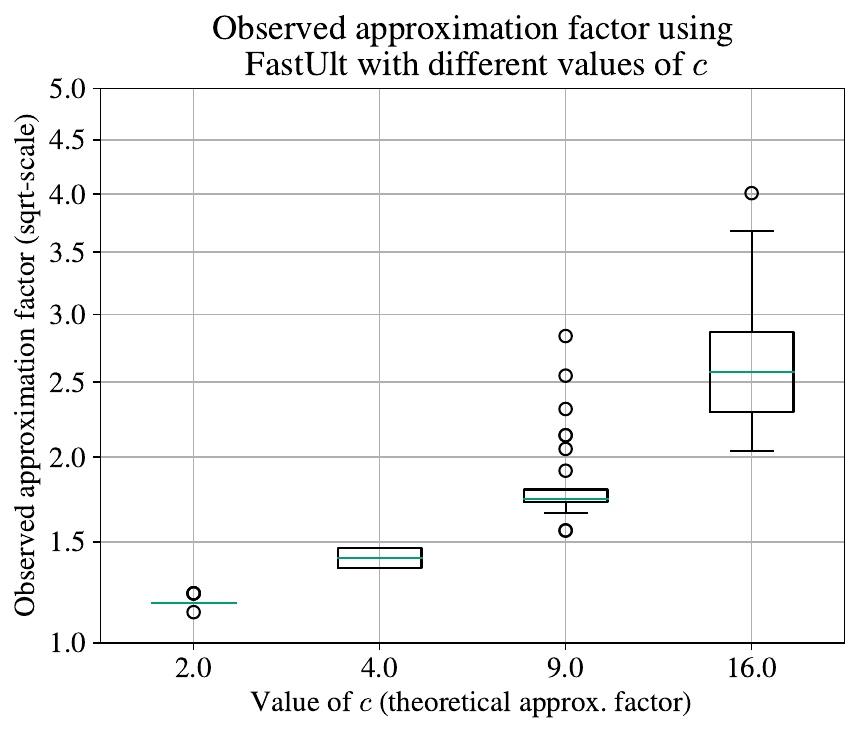}
    \caption{Approximation factor obtained by \fastbuf for different values of $c$.}
    \label{fig:buf_var}
\end{figure}

\paragraph{Experiment 2: Comparison with previous methods.}
We compare the performance of \fastbuf{} with the previous best known algorithm for \BUF{} by \citet{CDL21}, both in terms of the quality of the approximation factor of the best ultrametric and the running time. The results are given in \cref{table:cmp}. These results show that \fastbuf{} with $c = 4$ can achieve better ultrametric embeddings than the algorithm of \citet{CDL21} for a comparable computational cost. Furthermore, by using $c = 9$ or $16$, we can obtain embeddings with similar distortion but with a lower running time. Thus, another advantage of \fastbuf{} is the ability to trade off approximation factor and running time by adjusting the parameter $c$, for any desired embedding quality.

\paragraph{Experiment 3: Scaling.}

Finally, to evaluate how well our algorithm scales to larger datasets, we create two synthetic datasets with up to one million data points. These datasets contain $N$ uniformly random points from $[0, 1]^d$, for $(N,d) \in [(10^5, 100), (10^6, 10)]$. The running times are given in \cref{table:scale}.

\begin{table}[htbp]
    \centering
    \begin{tabular}{l|rr}
        \toprule
        & $N = 10^{5}$&$N = 10^{6}$ \\
         \multicolumn{1}{c|}{Algorithm} & $d = 100$&$d = 10$ \\
        \midrule
         \fastbuf($c = 4)$ & $1$m $39.89$s&$21$m $58.50$s\\
         \fastbuf($c = 9)$ & $34.07$s&$5$m $13.99$s\\
         \fastbuf($c = 16)$ & $19.24$s&$2$m $34.97$s\\
         CVL &  $7.96$s&  $1$m $54.48$s\\
         Single Linkage&  $9$m $55.36$s&  $\ge 10$h\\
        \bottomrule
    \end{tabular}
    \caption{Running time of the \fastbuf algorithm, the algorithm of \citet{CDL21} and the \textit{single linkage} algorithm from the \texttt{fastcluster} python package on datasets of $N$ $d$-dimensional random points. Each reported running time is an average of 30 runs.}
    \label{table:scale}
\end{table}

We observe that, while our algorithm is slower than that of \citet{CDL21}, \fastbuf does not suffer from the same quadratic blow-up as classical algorithms such as single linkage, and maintains a reasonable running time that can be afforded for the analysis of large datasets.

\subsection{Additional experiments}

\paragraph{Experiment A: Accuracy of the $\alpha$-approximate cut weights (ACW) algorithm.}
\label{par:experiment_A}
On \cref{line:acw} of the approximate cut weights algorithm (\cref{algo:cw}), we multiply the distance to the approximate farthest neighbor by $\alpha$ to ensure that $\ACW(e)$ is at least $\CutW(e)$ for every edge $e$.
However, our AFN data structure turns out to be very accurate in practice: it often returns points whose distance to the query point $q$ is close to maximal.
Therefore, the multiplication by $\alpha$ artificially increases the approximation factor of the algorithm\footnote{Recall that the approximation factor of the algorithm is the ratio of the distortion of the computed ultrametric divided by the distortion of the optimal ultrametric.}: multiplying by another number $c < \alpha$ would suffice in practice.

To quantify this phenomenon, we compute the approximation factor when multiplying by the smallest constant $c^*$ (the same for all queries) such that $\ACW(e)$ is at least $\CutW(e)$ for every $e$, instead of multiplying by $\alpha$.
To only measure the effect of $\alpha$-ACW algorithm on the approximation factor, we run this algorithm on an exact minimum spanning tree instead of a $\gamma$-KT.
The optimal distortion is computed using the algorithm of \citet{FKW95}.
The results are reported in \cref{fig:cw_var}.

\begin{figure}[htbp]
    \centering
    \includegraphics[width=0.7\textwidth]{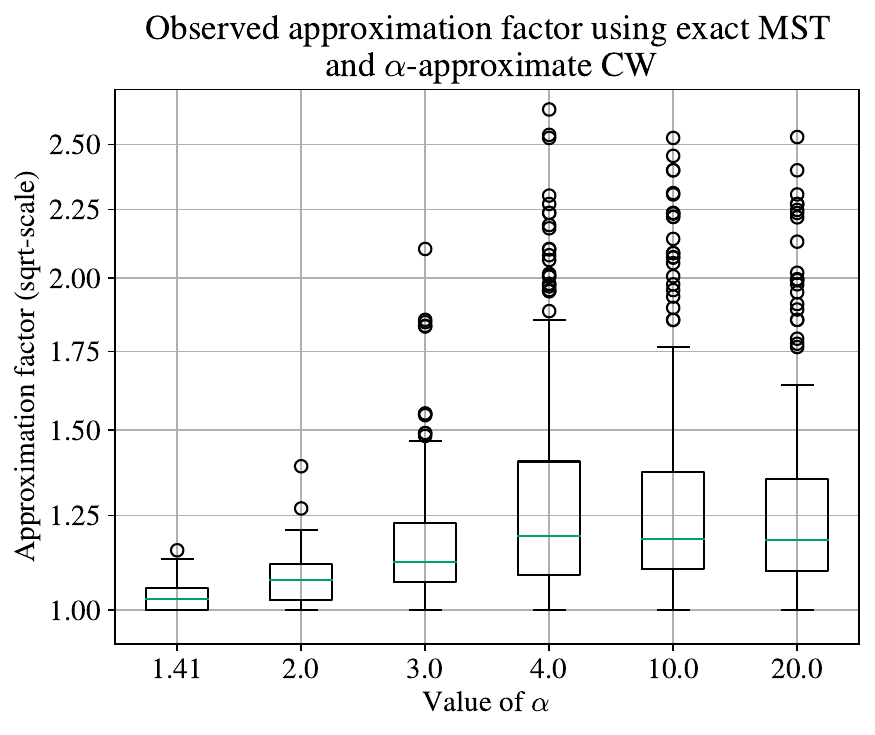}
    \caption{Accuracy of the $\alpha$-approximate cut weights algorithm. For each value of $\alpha$, the algorithm is run 30 times on each of the 5 datasets, resulting in 150 data points per boxplot.
    }
    \label{fig:cw_var}
\end{figure}

In practice, the observed approximation factor is much lower than the theoretical approximation factor: $75\%$ of the time, the observed approximation factor is less than $1.5$, and it is always less than $3$, even for $\alpha = 20$.
This means that, for $\alpha = 20$, the $\alpha$-ACW algorithm could multiply the distances by $3$ instead of $\alpha$ and still obtain an over-approximation of the cut weights.

\paragraph{Experiment B: Accuracy of the $\gamma$-KT algorithm.}
Next, we measure the accuracy of the $\gamma$-KT algorithm.
We run the exact cut weights algorithm on the output of the $\gamma$-KT algorithm, and compare the distortion of the resulting ultrametric to the optimal ultrametric (given by the algorithm of \citet{FKW95}).
The results are given in \cref{table:mst}: again we observe that the $\gamma$-KT algorithm performs much better than the theoretical guarantee.

\begin{table}[htbp]
    \centering
    \begin{tabular}{r|cc}
        \toprule
        \multicolumn{1}{c|}{MST} & Theoretical & Observed\\
        algorithm & approx. factor & approx. factor\\
        \midrule
        $1.41$-KT & $1.41$ & $1.06 \pm 0.08$ \\
        $2.0$-KT & $2.0$ & $1.23 \pm 0.23$ \\
        $3.0$-KT & $3.0$ & $1.63 \pm 0.47$ \\
        $4.0$-KT & $4.0$ & $1.87 \pm 0.64$ \\
        $10.0$-KT & $10.0$ & $3.37 \pm 2.21$ \\
        $20.0$-KT & $20.0$ & $5.22 \pm 5.57$ \\
        \bottomrule
    \end{tabular}
    \caption{Accuracy of the $\gamma$-KT algorithm with exact cut weights. For each value of $\gamma$, 
        the average and standard deviation of the approximation factor are computed over
        $30$ runs of the algorithm on each of the $5$ datasets, for a total of $150$ data points.}
    \label{table:mst}
\end{table}


\bibliographystyle{plainnat}
\bibliography{biblio}

\begin{thebibliography}{25}
\providecommand{\natexlab}[1]{#1}
\providecommand{\url}[1]{\texttt{#1}}
\expandafter\ifx\csname urlstyle\endcsname\relax
  \providecommand{\doi}[1]{doi: #1}\else
  \providecommand{\doi}{doi: \begingroup \urlstyle{rm}\Url}\fi

\bibitem[Agarwala et~al.(1998)Agarwala, Bafna, Farach, Paterson, and
  Thorup]{agarwala1998approximability}
Richa Agarwala, Vineet Bafna, Martin Farach, Mike Paterson, and Mikkel Thorup.
\newblock On the approximability of numerical taxonomy (fitting distances by
  tree metrics).
\newblock \emph{SIAM Journal on Computing}, 28\penalty0 (3):\penalty0
  1073--1085, 1998.

\bibitem[Ailon and Charikar(2011)]{ailon2011fitting}
Nir Ailon and Moses Charikar.
\newblock Fitting tree metrics: Hierarchical clustering and phylogeny.
\newblock \emph{SIAM Journal on Computing}, 40\penalty0 (5):\penalty0
  1275--1291, 2011.

\bibitem[Andoni and Indyk(2006)]{andoni2006near}
Alexandr Andoni and Piotr Indyk.
\newblock Near-optimal hashing algorithms for approximate nearest neighbor in
  high dimensions.
\newblock In \emph{Proceedings of the 47th Annual IEEE Symposium on Foundations
  of Computer Science}, pages 459--468, 2006.

\bibitem[Andoni and Zhang(2023)]{andoni2023sub}
Alexandr Andoni and Hengjie Zhang.
\newblock Sub-quadratic (1+$\epsilon$)-approximate euclidean spanners, with
  applications.
\newblock \emph{arXiv preprint arXiv:2310.05315}, 2023.

\bibitem[Carlsson and M{\'{e}}moli(2010)]{CarlssonM10}
Gunnar~E. Carlsson and Facundo M{\'{e}}moli.
\newblock Characterization, stability and convergence of hierarchical
  clustering methods.
\newblock \emph{J. Mach. Learn. Res.}, 11:\penalty0 1425--1470, 2010.
\newblock \doi{10.5555/1756006.1859898}.
\newblock URL \url{https://dl.acm.org/doi/10.5555/1756006.1859898}.

\bibitem[Charikar and Chatziafratis(2017)]{charikar2017approximate}
Moses Charikar and Vaggos Chatziafratis.
\newblock Approximate hierarchical clustering via sparsest cut and spreading
  metrics.
\newblock In \emph{Proceedings of the Twenty-Eighth Annual ACM-SIAM Symposium
  on Discrete Algorithms}, pages 841--854. SIAM, 2017.

\bibitem[Cochez and Mou(2015)]{cochez2015twister}
Michael Cochez and Hao Mou.
\newblock Twister tries: Approximate hierarchical agglomerative clustering for
  average distance in linear time.
\newblock In \emph{Proceedings of the 2015 ACM SIGMOD international conference
  on Management of data}, pages 505--517, 2015.

\bibitem[Cohen-Addad et~al.(2017)Cohen-Addad, Kanade, and
  Mallmann-Trenn]{cohen2017hierarchical}
Vincent Cohen-Addad, Varun Kanade, and Frederik Mallmann-Trenn.
\newblock Hierarchical clustering beyond the worst-case.
\newblock \emph{Advances in Neural Information Processing Systems}, 30, 2017.

\bibitem[Cohen-Addad et~al.(2019)Cohen-Addad, Kanade, Mallmann-Trenn, and
  Mathieu]{cohen2019hierarchical}
Vincent Cohen-Addad, Varun Kanade, Frederik Mallmann-Trenn, and Claire Mathieu.
\newblock Hierarchical clustering: Objective functions and algorithms.
\newblock \emph{Journal of the ACM (JACM)}, 66\penalty0 (4):\penalty0 1--42,
  2019.

\bibitem[Cohen{-}Addad et~al.(2020)Cohen{-}Addad, {Karthik {C. S.}}, and
  Lagarde]{CKL20}
Vincent Cohen{-}Addad, {Karthik {C. S.}}, and Guillaume Lagarde.
\newblock On efficient low distortion ultrametric embedding.
\newblock \emph{CoRR}, abs/2008.06700, 2020.
\newblock URL \url{https://arxiv.org/abs/2008.06700}.

\bibitem[Cohen-Addad et~al.(2021)Cohen-Addad, de~Joannis~de Verclos, and
  Lagarde]{CDL21}
Vincent Cohen-Addad, Rémi de~Joannis~de Verclos, and Guillaume Lagarde.
\newblock Improving ultrametrics embeddings through coresets.
\newblock In \emph{Proceedings of the 38th International Conference on Machine
  Learning, {ICML} 2021, 18-24 July 2021, Virtual Event}, volume 139 of
  \emph{Proceedings of Machine Learning Research}, pages 2060--2068. {PMLR},
  2021.
\newblock URL \url{http://proceedings.mlr.press/v139/cohen-addad21a.html}.

\bibitem[Dasgupta(2016)]{dasgupta2016cost}
Sanjoy Dasgupta.
\newblock A cost function for similarity-based hierarchical clustering.
\newblock In \emph{Proceedings of the forty-eighth annual ACM symposium on
  Theory of Computing}, pages 118--127, 2016.

\bibitem[Farach et~al.(1995)Farach, Kannan, and Warnow]{FKW95}
Martin Farach, Sampath Kannan, and Tandy~J. Warnow.
\newblock A robust model for finding optimal evolutionary trees.
\newblock \emph{Algorithmica}, 13\penalty0 (1/2):\penalty0 155--179, 1995.
\newblock \doi{10.1007/BF01188585}.
\newblock URL \url{https://doi.org/10.1007/BF01188585}.

\bibitem[Gilpin et~al.(2013)Gilpin, Qian, and Davidson]{gilpin2013efficient}
Sean Gilpin, Buyue Qian, and Ian Davidson.
\newblock Efficient hierarchical clustering of large high dimensional datasets.
\newblock In \emph{Proceedings of the 22nd ACM international conference on
  Information \& Knowledge Management}, pages 1371--1380, 2013.

\bibitem[Har-Peled et~al.(2013)Har-Peled, Indyk, and
  Sidiropoulos]{har2013euclidean}
Sariel Har-Peled, Piotr Indyk, and Anastasios Sidiropoulos.
\newblock Euclidean spanners in high dimensions.
\newblock In \emph{Proceedings of the twenty-fourth annual ACM-SIAM symposium
  on Discrete algorithms}, pages 804--809. SIAM, 2013.

\bibitem[Jain and Dubes(1988)]{jain1988algorithms}
Anil~K Jain and Richard~C Dubes.
\newblock \emph{Algorithms for clustering data}.
\newblock Prentice-Hall, Inc., 1988.

\bibitem[Johnson et~al.(1986)Johnson, Lindenstrauss, and
  Schechtman]{johnson1986extensions}
William~B Johnson, Joram Lindenstrauss, and Gideon Schechtman.
\newblock Extensions of lipschitz maps into banach spaces.
\newblock \emph{Israel Journal of Mathematics}, 54\penalty0 (2):\penalty0
  129--138, 1986.

\bibitem[Kelly et~al.(2024{\natexlab{a}})Kelly, Longjohn, and
  Nottingham]{kagglediabetes}
Markelle Kelly, Rachel Longjohn, and Kolby Nottingham.
\newblock Diabetes dataset - kaggle.
\newblock
  \url{https://www.kaggle.com/datasets/uciml/pima-indians-diabetes-database},
  2024{\natexlab{a}}.
\newblock Accessed: 2024-06-23.

\bibitem[Kelly et~al.(2024{\natexlab{b}})Kelly, Longjohn, and
  Nottingham]{uci_repo}
Markelle Kelly, Rachel Longjohn, and Kolby Nottingham.
\newblock The {UCI} machine learning repository.
\newblock \url{http://archive.ics.uci.edu}, 2024{\natexlab{b}}.
\newblock URL \url{https://archive.ics.uci.edu}.
\newblock Accessed: 2024-06-23.

\bibitem[Kruskal(1956)]{kruskal1956shortest}
Joseph~B Kruskal.
\newblock On the shortest spanning subtree of a graph and the traveling
  salesman problem.
\newblock \emph{Proceedings of the American Mathematical society}, 7\penalty0
  (1):\penalty0 48--50, 1956.

\bibitem[Moseley and Wang(2023)]{moseley2023approximation}
Benjamin Moseley and Joshua~R Wang.
\newblock Approximation bounds for hierarchical clustering: Average linkage,
  bisecting k-means, and local search.
\newblock \emph{Journal of Machine Learning Research}, 24\penalty0
  (1):\penalty0 1--36, 2023.

\bibitem[Murtagh and Contreras(2017)]{MurtaghC17}
Fionn Murtagh and Pedro Contreras.
\newblock Algorithms for hierarchical clustering: an overview, {II}.
\newblock \emph{WIREs Data Mining Knowl. Discov.}, 7\penalty0 (6), 2017.
\newblock \doi{10.1002/WIDM.1219}.
\newblock URL \url{https://doi.org/10.1002/widm.1219}.

\bibitem[Pagh et~al.(2017)Pagh, Silvestri, Sivertsen, and
  Skala]{pagh2017approximate}
Rasmus Pagh, Francesco Silvestri, Johan Sivertsen, and Matthew Skala.
\newblock Approximate furthest neighbor with application to annulus query.
\newblock \emph{Information Systems}, 64:\penalty0 152--162, 2017.

\bibitem[Roy and Pokutta(2017)]{roy2017hierarchical}
Aurko Roy and Sebastian Pokutta.
\newblock Hierarchical clustering via spreading metrics.
\newblock \emph{Journal of Machine Learning Research}, 18\penalty0
  (88):\penalty0 1--35, 2017.

\bibitem[Wareham(1993)]{wareham1993complexity}
HT~Wareham.
\newblock On the complexity of inferring evolutionary trees.
\newblock \emph{Technical Report Technical Report}, 9301, 1993.

\end{thebibliography}

\end{document}